\documentclass[journal]{IEEEtran}

\usepackage{cite}
\usepackage{amsmath,amssymb,amsfonts,amsthm}
\usepackage{algorithmic}
\usepackage{graphicx}
\usepackage{textcomp}
\usepackage{xcolor}
\usepackage{bm}        
\usepackage{booktabs}  
\usepackage{subfig}
\usepackage{subcaption}
\usepackage{physics}
\usepackage{hyperref}
\hypersetup{
    colorlinks=true,    
    citecolor=blue,     
    linkcolor=blue,     
    urlcolor=blue       
}
\newtheorem{theorem}{Theorem}
\newtheorem{lemma}{Lemma}

\newtheorem{assumption}{Assumption}
\newtheorem{remark}{Remark}
\newtheorem{proposition}{Proposition}

\newcommand*{\vecbf}[1]{\mathbf{#1}} 
\newcommand*{\matbf}[1]{\mathbf{#1}} 
\newcommand*{\gbf}[1]{\bm{#1}} 
\newcommand*{\myset}[1]{\mathcal{#1}} 

\newcommand{\R}{\mathbb{R}}
\newcommand{\F}{\mathbf{F}}
\newcommand{\p}{\mathbf{p}}
\newcommand{\uu}{\mathbf{u}}
\newcommand{\q}{\mathbf{q}}
\newcommand{\SO}{\mathbf{SO}}

\newcommand{\Proj}{\mathrm{Proj}}
\newcommand{\M}{\mathcal{M}}
\newcommand{\Q}{\mathcal{Q}}
\newcommand{\Span}{\mathrm{span}}

\begin{document}

\title{Element-based Formation Control: a Unified Perspective from Continuum Mechanics}

\author{Kun Cao, \IEEEmembership{Member, IEEE}, and Lihua Xie, \IEEEmembership{Fellow, IEEE}%
\thanks{This work was supported by [Grant Number].}%
\thanks{K. Cao (corresponding author) is with the Department of Control Science and Engineering, College of Electronics and Information Engineering, Tongji University, Shanghai 201804, China,  the Shanghai Institute of Intelligent Science and Technology, National Key Laboratory of Autonomous Intelligent Unmanned Systems, and Frontiers Science Center for Intelligent Autonomous Systems, Ministry of Education, Beijing 100816, China~{\tt \footnotesize caokun@tongji.edu.cn}\\
L.~Xie is with the School of Electrical and Electronic Engineering, Nanyang Technological University, 50 Nanyang Avenue, Singapore 639798~{\tt \footnotesize elhxie@ntu.edu.sg}\\
}}
\maketitle

\begin{abstract}
This paper establishes a unified element-based framework for formation control by introducing the concept of the deformation gradient from continuum mechanics.
Unlike traditional methods that rely on geometric constraints defined on graph edges, we model the formation as a discrete elastic body composed of simplicial elements. 
By defining a generalized distortion energy based on the local deformation gradient tensor, we derive a family of distributed control laws that can enforce various geometric invariances, including translation, rotation, scaling, and affine transformations. 
The convergence properties and the features of the proposed controllers are analyzed in detail. 
Theoretically, we show that the proposed framework serves as a bridge between existing rigidity-based and Laplacian-based approaches. Specifically, we show that rigidity-based controllers are mathematically equivalent to minimizing specific projections of the deformation energy tensor. 
Furthermore, we establish a rigorous link between the proposed energy minimization and Laplacian-based formation control. Numerical simulations in 2D and 3D validate the effectiveness and the unified nature of the proposed framework.
\end{abstract}


\section{Introduction}
\label{sec:intro}
\IEEEPARstart{M}{ulti-agent} formation control has garnered significant research attention in recent decades due to its broad applications in satellite clustering, cooperative surveillance, and swarming robotics~\cite{liu2024survey,su2026bearing}. 
The fundamental objective is to drive a group of agents to achieve a desired geometric shape while maintaining specific invariances for maneuvering when encountering complex environments.

For decades, the field has been bifurcated into two dominant paradigms: rigidity-based methods and Laplacian-based methods.
Rigidity-based approaches focus on sparse, edge-level geometric constraints, depending on the sensed variables, these include displacement-based control for translation-invariant formations \cite{oh2015survey,olfati2004consensus,ji2007distributed}, distance-based control for rigid formations \cite{sun2017distributed,sun2018conservation,DBLP:journals/automatica/KwonSAA22}, bearing-based control for scale-invariant formations \cite{zhao2019bearing,DBLP:journals/tac/TrinhTA20}, and angle-based or ratio-of-distance (RoD)-based control for similarity-invariant formations \cite{chen2019anglerigidityusagestabilize,jing2019angle, cao2019ratio,chen2025angle}. 
While effective, these strategies often treat geometric constraints in isolation, lacking a unified physical interpretation that links local interactions to the holistic deformation of the formation.

On the other hand, Laplacian-based methods, which use barycentric coordinates and stress matrices to construct the signed Laplacian \cite{lin2016distributed,lin2016graph,lin2016necessary,han2018barycentric,zhao2018affine,yang2018constructing,DBLP:journals/tac/LiuCWCQL26}, provide a powerful framework for formation control. 
These methods rely on the global properties of the Laplacian matrix (e.g., kernel space and positive definiteness) to stabilize formations and usually lead to global convergence results.
However, as we categorized above as two seemingly irrelevant streams, the underlying connection between the geometric constraints in rigidity-based methods and the graph properties in Laplacian-based methods remains implicit and unexplored. 

In parallel, the computer graphics communities have developed robust methods for physics-based simulations, such as As-Rigid-As-Possible (ARAP)~\cite{sorkine2007rigid} and As-Similar-As-Possible (ASAP) surface modeling. 
Rooted in continuum mechanics~\cite{bonet_wood_1997}, these methods minimize a distortion energy defined on mesh elements (triangles or tetrahedra) to simulate physical behaviors. 
Although these techniques are highly successful in their domain, they typically focus on solving offline optimization problems for the next state, rather than designing distributed feedback control laws for dynamical agents. 
Some recent works have attempted to apply continuum mechanics to multi-agent systems. For instance,~\cite{Rastgoftar2016continuum} employs homogeneous transformations for leader-based deformation, but the followers still rely on the Laplacian matrix. 
Others, such as~\cite{freudenthaler2021formation}, derive control laws from continuum models by taking the limit as the number of agents approaches infinity (PDE-based control), which does not focus on the discrete nature of inter-agent measurements and specific geometric invariances in finite networks.

Motivated by these observations, this paper proposes a fundamental paradigm shift: modeling the multi-agent formation not as a sparse collection of edges, but as a discrete elastic body governed by continuum mechanics. 
By elevating the fundamental unit of interaction from one-dimensional graph edges to multidimensional simplicial elements—such as triangles in 2D or tetrahedra in 3D—we introduce the deformation gradient as the primary descriptor of formation error. 
This approach moves beyond the ``patchwork" of edge-based constraints and instead treats the formation as a coherent material ensemble endowed with intrinsic geometric fidelity.

The conceptual elegance of this framework lies in its ability to serve as a theoretical bridge between the long-standing dichotomy of rigidity-based and Laplacian-based methods. 
We demonstrate that traditional rigidity constraints—such as distances, bearings, and angles—emerge naturally as specific sparse projections of the dense deformation energy tensor. 
Furthermore, by analyzing the Dirichlet energy of the deformation gradient, we reveal a rigorous mathematical link to Laplacian-based control. 
This unification provides a first-principle perspective on formation control, showing that disparate methods are, in fact, localized instantiations of a universal energy minimization principle.
The main contributions of this work are twofold:
\begin{enumerate}
    \item We propose a generalized energy minimization framework based on the local deformation gradient. 
    By instantiating specific energy density functions, our framework systematically derives controllers that enforce translation, rotation, scaling, and similarity invariances within a single, consistent codebase.
    \item We provide a comprehensive mapping of existing rigidity-based and Laplacian-based methods into our element-based framework.
    Specifically, we prove that rigidity-based constraints are sparse samplings of the deformation tensor, while Laplacian-based methods correspond to the minimization of its Dirichlet energy.
\end{enumerate}

The remainder of this paper is organized as follows. Section~\ref{sec:problem} formulates the problem. Section~\ref{sec:element} details the element-based controller design. 
Section~\ref{sec:theo} provides the theoretical analysis and connections with existing methods. 
Simulation results are presented in Sec.~\ref{sec:simulation}, followed by conclusions in Sec.~\ref{sec:conclusion}.

\emph{Notations:} In this paper, $a$, $\vecbf{a}$, $\matbf{a}$, and $\myset{A}$ denote the scalar, vector, matrix, and set, respectively. 
Let $\|\vecbf{a}\|$, $\|\vecbf{A}\|_{F}$ denote the $2$-norm of vector and Frobenius norm of matrix. 
Denote by $\matbf{A}^{\top}$ and $\matbf{A}^{-1}$ the transpose and inverse of $\matbf{A} \in \mathbb{R} ^{n \times n}$, respectively. 
Let $\matbf{I}_{n} \in \mathbb{R}^{n \times n}$ be the $n$-dimensional identity matrix and $\vecbf{1}_{n}$ be the $n$-dimensional column vector with all entries of $1$. 
Let $\odot$ denote the tensor contraction operation.


\section{Problem Formulation}
\label{sec:problem}

\begin{figure}
    \centering
    \includegraphics[width=\linewidth]{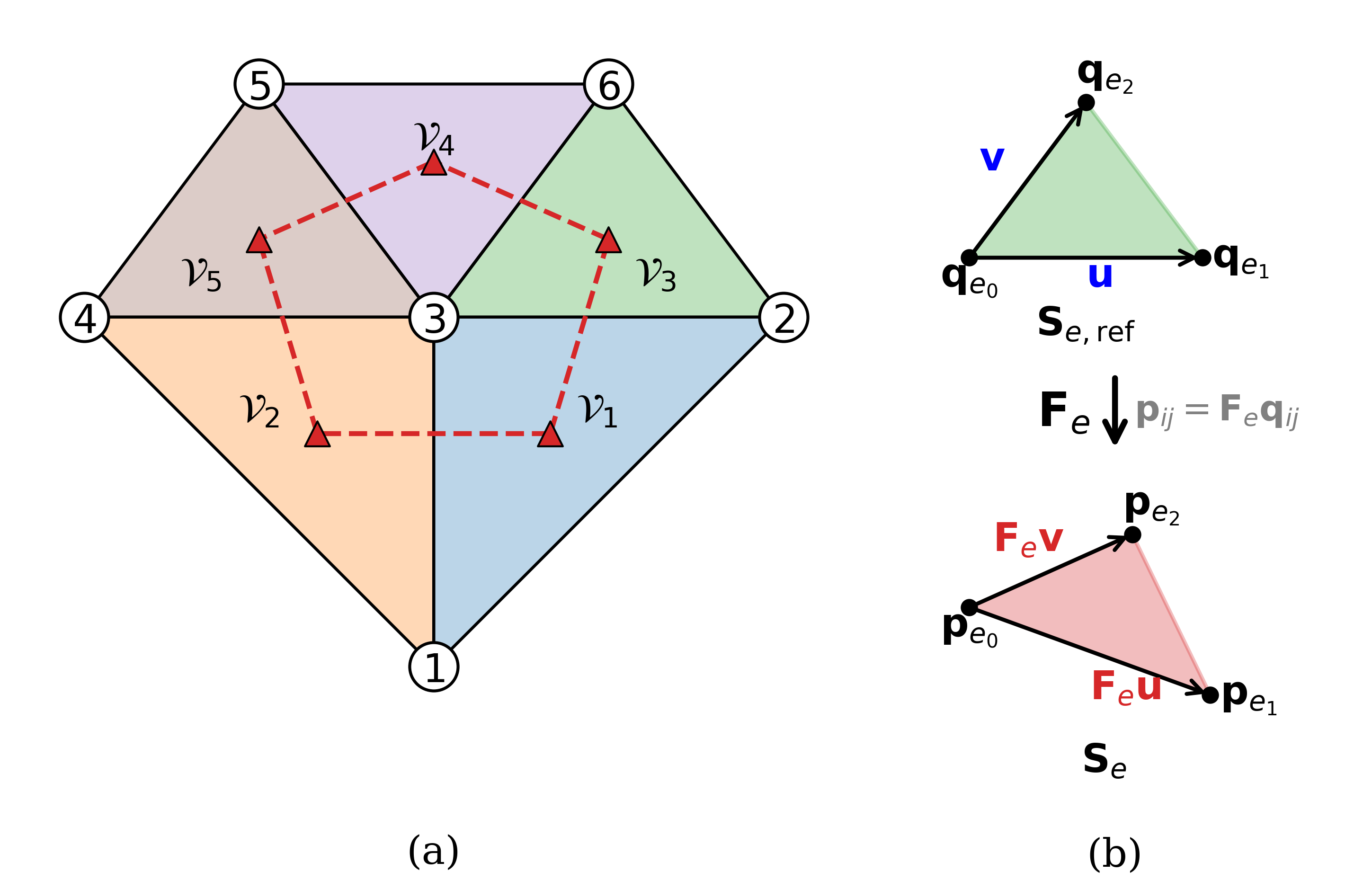}
    \caption{(a) Simplicial complex representation of the multi-agent formation. 
    The numbered circles and solid black lines denote the physical agents and the 1D communication topology, respectively. 
    The macroscopic formation is tiled by five 2D simplicial elements (colored regions). 
    The red triangular nodes ($\mathcal{V}_e$) and dashed lines illustrate the corresponding dual graph, which models the adjacency and interaction pathways among neighboring elements. 
    (b) Illustration of the deformation gradient mapping a reference element to its current configuration in 2D space.}
    \label{fig:combined_fig1}
\end{figure}

Consider a multi-agent system consisting of $N$ agents moving in a $d$-dimensional space ($d=2$ or $3$). 
Let $\p_i \in \R^d$ denote the position of agent $i$, and $\p = [\p_1^{\top}, \dots, \p_N^{\top}]^{\top} \in \R^{dN}$ denote the configuration of the entire formation.

Traditional rigidity-based approaches developed various graphical conditions for rigidity \cite{zhao2019bearing,cao2019ratio,chen2025angle,jing2019angle} over a notion of framework $\mathcal{F} = (\mathcal{G}, \p)$, where $\mathcal{G} = (\mathcal{V}, \mathcal{H})$ is the underlying measurement/communication graph.
Here, $\mathcal{V} = \{1, \dots, N\}$ is the set of agents and $\mathcal{H} \subseteq \mathcal{V} \times \mathcal{V}$ represents the measurement/communication links, as illustrated by the numbered circles and solid black lines in Fig. \ref{fig:combined_fig1}(a).

In contrast, we assume that our formation is defined on a \textit{simplicial complex} (e.g., triangles in 2D, tetrahedra in 3D), which is the basic notion of element in continuum mechanics \cite{bonet_wood_1997}.
Specifically, a $d$-dimensional element (simplex) $e$ is formed by $d+1$ fully connected agents.
Let $\mathcal{V}_e = \{e_0, \dots, e_d\} \subseteq \mathcal{V}$ denote the vertex set of element $e$, and $\mathbf{P}_e = [\mathbf{p}_{e_0}, \dots, \mathbf{p}_{e_d}]$ denote its corresponding configuration set.
We define $\mathcal{E}$ as the collection of all such elements tiling the formation (the colored regions in Fig. \ref{fig:combined_fig1}(a)). 

To facilitate the coordination among adjacent elements, we redefine the formation architecture via a dual graph $\hat{\mathcal{G}} = (\hat{\mathcal{V}}, \hat{\mathcal{H}})$. 
As depicted in Fig. \ref{fig:combined_fig1}(a), the dual node set $\hat{\mathcal{V}} = \{\mathcal{V}_1, \dots, \mathcal{V}_{|\mathcal{E}|}\}$ represents the elements themselves (indicated by the red triangular nodes), and the dual edge set $\hat{\mathcal{H}}$ represents the links between two elements which share $(d-1)$-dimensional faces (the red dashed lines).

Let $\q = [\q_1^{\top}, \dots, \q_N^{\top}]^{\top} \in \R^{dN}$ be the nominal or reference configuration (the desired shape). 
We consider the following single-integrator dynamics for each agent:
\begin{equation} \label{eq:dyn}
    \dot{\p}_i = \uu_{i}, \quad i = 1, \dots, N,
\end{equation}
where $\uu_i \in \R^d$ is the control input for agent $i$.  
The objective is to drive the agents to a configuration $\p$ that is equivalent to $\q$ under each of the following transformations:
\begin{equation} \label{eq:M}
    \begin{aligned}
        & \mathrm{(Translation)} & \M_{\mathrm{T}} & = \left\{\p \in \mathbb{R}^{Nd} \mid \p_i = \q_i + \mathbf{t}, \right. \\
        & & & \left. \mathbf{t} \in \R^d, \forall i \right\} \\
        & \mathrm{(Rotation)} & \M_{\mathrm{TR}} & = \left \{\p \in \mathbb{R}^{Nd} \mid \p_i = \mathbf{R}\q_i + \mathbf{t},  \right. \\
        & & & \left. \mathbf{R} \in \SO(d), \mathbf{t} \in \R^d, \forall i \right \} \\
        & \mathrm{(Scaling)} & \M_{\mathrm{TS}} & = \left\{\p \in \mathbb{R}^{Nd} \mid \p_i = s \q_i + \mathbf{t}, \right.  \\
        & & & \left. s \in \R, \mathbf{t} \in \R^d, \forall i \right\}  \\
        & \mathrm{(Similarity)} & \M_{\mathrm{TRS}} & = \left\{\p \in \mathbb{R}^{Nd} \mid \p_i = s \mathbf{R}\q_i + \mathbf{t}, \right.  \\
        & & &  \hspace{-0.5cm}\left. s \in \R^{+}, \mathbf{R} \in \SO(d), \mathbf{t} \in \R^d, \forall i \right\}.
    \end{aligned}
\end{equation}

To ensure well-posedness of the problem, we make the following standard assumptions.
\begin{assumption}[Non-degenerate Reference]
    \label{ass:non-degenerate}
    For every simplex $e \in \mathcal{E}$, the vertices in the reference configuration $\{\q_j\}_{j \in \mathcal{V}_e}$ are affinely independent.
\end{assumption}

\begin{assumption}[Topology Connectivity]
    \label{ass:topo}
    The dual graph $\hat{\mathcal{G}}$ is connected. 
\end{assumption}

Assumption \ref{ass:non-degenerate} is standard in continuum mechanics to ensure well-defined deformation gradients (will be introduced in the sequel); it is also standard in rigidity-based literature to avoid singularities for rigidity matrices~\cite{zhao2019bearing,cao2019ratio,jing2019angle}.
Assumption \ref{ass:topo} guarantees that the formation behaves as a single connected body, preventing isolated sub-formations.
It is implicitly assumed that the reference configuration $\q$ is known in the problem formulation. 
However, it is also valid if only a set of partial information, e.g., distance, bearing, etc., is available. 
We defer the discussion on this to Sec. \ref{sec:theo} to avoid distraction.

\section{Element-based Controller Design}
\label{sec:element}
In this section, we shall present our element-based formation controller design.
Before that, we introduce the concept of the discrete deformation gradient for each element in continuum mechanics. 
Then, we can construct various energy functions based on this and derive our controllers.

\subsection{Discrete Deformation Gradient}
Discrete deformation gradient, originally defined in the finite strain theory in continuum mechanics, describes the relationship between the reference and current configuration and expresses motion locally around a point. 
If we view our formation as an elastic body, and the formation element as a single element in the finite element method, we can define this concept similarly for our formation element. 
In particular, for a specific element $e$, let vertex $e_{0}$ be the local origin, define the reference and the current shape matrix as:
\begin{equation*}
    \begin{aligned}
        \mathbf{S}_{e,\mathrm{ref}} &= [\q_{e_{1}} - \q_{e_{0}}, \dots, \q_{e_{d}} - \q_{e_{0}}],  \\
        \mathbf{S}_{e} &= [\p_{e_{1}} - \p_{e_{0}}, \dots, \p_{e_{d}} - \p_{e_{0}}]. 
    \end{aligned}
\end{equation*}
The deformation gradient $\F_e \in \R^{d \times d}$ maps the reference vectors to the current spatial vectors, i.e., $\F_e \mathbf{S}_{e,\mathrm{ref}} = \mathbf{S}_{e}$, hence
\begin{equation} \label{eq:def_F}
    \F_e = \mathbf{S}_{e} \mathbf{S}_{e,\mathrm{ref}}^{-1},
\end{equation}
where $\mathbf{S}_{e,\mathrm{ref}}^{-1}$ is well-defined in view of Assumption \ref{ass:non-degenerate}. 
An illustration of the deformation gradient can be seen from Fig. \ref{fig:combined_fig1}(b).
It can be found from the definition that $\F_e$ captures all local geometric distortions, including rotation, stretch, and shear, which suggests that this is a good option for characterizing the formation error.

\subsection{Controller Design}
In this section, we shall see how to design controllers based on the discrete deformation gradient $\F_e$ introduced above. 
Firstly, we consider the following total distortion energy function $V(\p)$, which is the sum of local energy densities $\Psi(\F_e)$ defined on each element:
\begin{equation} \label{eq:energy}
    V(\p) = \sum_{e \in \mathcal{E}} w_e \cdot \Psi(\F_e),
\end{equation}
where $w_e$ is the volume of the reference element, i.e., $w_e = |\det(\mathbf{S}_{e,\mathrm{ref}})|/d!$ \footnote{From a pure optimization-based control perspective, any strictly positive weights $w_e > 0$ preserve the convexity of the local energy and guarantee convergence, as will be shown in Sec. \ref{subsec:conv}. 
However, choosing $w_e$ as the element's geometric volume ensures exact consistency with the volume integration of strain energy in continuum mechanics. 
Furthermore, as will be discussed in Sec. \ref{subsec:con}, these weights can be treated as tunable parameters (or even take negative values) to recover specific Laplacian matrix properties.}.
It should be noted that this form is highly general since by choosing different $\Psi(\F)$ (the subscript $(\cdot)_{e}$ will be omitted in the sequel if no confusion will be caused), we can instantiate the following different types of energy functions and derive controllers respectively. 
\subsubsection{Translation-invariant distortion energy}
It can be found that any $\F \neq \mathbf{I}$ will cause a distortion of the shape of the element, while $\F = \mathbf{I}$ means that the current element is exactly a translation of the reference element.
Based on this observation, we can define the following translation-invariant distortion energy function
\begin{equation}
    \label{eq:energy_t}
    \Psi_{\mathrm{T}}(\F) = \|\F - \mathbf{I}\|_F^2.
\end{equation}
This penalizes any deviation from a pure translation.

\subsubsection{Rotation-invariant distortion energy}
From the point view of SVD decomposition, any square matrix $\mathbf{F}$ can be decomposed as $\mathbf{F} = \mathbf{U} \bm{\Sigma} \mathbf{V}^{\top}$, where $\mathbf{U}, \mathbf{V} \in \SO(d)$ and $\bm{\Sigma} = \mathrm{diag}(\sigma_1, \dots, \sigma_d)$ with singular values $\sigma_1, \dots, \sigma_d > 0$. 
In view of this, we can see that $(\mathbf{U}, \mathbf{V}^{\top})$ represents the rotational part of the deformation while $\bm{\Sigma}$ represents the pure stretch part.
Based on this observation, we can define the following rotation-invariant distortion energy function
\begin{equation} \label{eq:energy_tr}
    \Psi_{\mathrm{TR}}(\F) = \|\F - \mathbf{R}\|_F^2,
\end{equation}
where $\mathbf{R} = \mathbf{U} \mathbf{V}^{\top}$.
This penalizes any deviation from a combination of translation and rotation.

\subsubsection{Scaling-invariant distortion energy}
As a `counterpart' to the above, instead of penalizing the deviation from the rotation part, we can also penalize the deviation from a pure scaling.
In particular, we can define the following scaling-invariant distortion energy function
\begin{equation} \label{eq:energy_ts}
    \Psi_{\mathrm{TS}}(\F) = \|\F -  s_{\mathrm{TS}}\mathbf{I}\|_F^2,
\end{equation}
where $s_{\mathrm{TS}} = \trace(\F)/d$ is the average scaling factor.

\subsubsection{Similarity-invariant distortion energy}
As the final type of energy function which includes all the above invariances, we consider the following definition:
\begin{equation} \label{eq:energy_trs}
    \Psi_{\mathrm{TRS}}(\F) = \|\F - s_{\mathrm{TRS}} \mathbf{R}\|_F^2,
\end{equation}
where $s_{\mathrm{TRS}} = \trace(\mathbf{R}^{\top}\F)/d$ and $\mathbf{R}$ is defined as above.

These choices of regressor (i.e., the second term inside the Frobenius norm) will be evident in the subsequent design.
Here, we first present the general distributed control law based on the above energy functions, as has usually been done in rigidity-based controllers.
Specifically, we define the control input for agent $i$ as the negative gradient of the total energy w.r.t. its position, i.e.,
\begin{equation} \label{eq:control_law}
    \uu_i = - \nabla_{\p_i} V(\p) = - \sum_{e \in \mathcal{E}} w_{e} \frac{\partial \Psi(\F_e)}{\partial \F_e} \frac{\partial \F_e}{\partial \p_i}.
\end{equation}
In the above, $\frac{\partial \F_e}{\partial \p_i}$ is the derivative of the deformation gradient with respect to agent $i$'s position.
For agent $e_{0}$ acting as the base vertex of element $e$, the derivative is
\begin{equation*}
    \frac{\partial \F_e}{\partial \p_{e_{0}}} = -\mathbf{I}  \cdot \mathbf{S}_{e,\mathrm{ref}}^{-1}.
\end{equation*}
For agent $i$ acting as vertex $e_{j}$ in element $e$, $j = 1, \dots, d$, the derivative is
\begin{equation*}
    \frac{\partial \F_e}{\partial \p_{e_{j}}} = \mathbf{E}_j \cdot \mathbf{S}_{e,\mathrm{ref}}^{-1},
\end{equation*}
where $\mathbf{E}_j$ is the $j$-th standard basis vector.
The other term $\frac{\partial \Psi(\F_e)}{\partial \F_e}$ is the gradient of the energy function w.r.t. the deformation gradient, which, borrowing terminology from continuum mechanics, is called the first Piola-Kirchhoff stress\cite{bonet_wood_1997}. 
In our context, this is simply the partial derivative serving as an intermediate quantity in the gradient computation with no independent physical meaning beyond that.
To facilitate the clarity of presentation, we first use a unified notation for the above energy functions
\begin{equation}
    \Psi(\F) = \|\F - \Proj(\F)\|_F^2,
\end{equation}
where $\Proj(\F)$ is self-evident from the above and can be written as the optimal point of optimization problems detailed in Appendix \ref{app:proj}.
With these definitions, one has
\begin{equation*}
    \begin{aligned}
            \frac{\partial \Psi(\F)}{\partial \F} & = 2(\F - \Proj(\F)) - 2 \cdot \frac{\partial \Proj(\F)}{\partial \F} \odot (\F - \Proj(\F)) \\ 
            & = 2(\F - \Proj(\F)),
    \end{aligned}
\end{equation*}
where the second equality follows from the Envelope theorem, and the detailed proof can be found in Appendix \ref{app:grad}.
This simplification explains the choice of the regressor in the above, since closed-form expressions may not be obtained otherwise. 

This control law relies only on the measurement of relative positions of neighbors in each element; we call it an element-based controller, and it is distributed by design. 

\subsection{Further Discussions}

As we mentioned throughout the paper, the concept of deformation gradient is originally from continuum mechanics, which describes the local motion of a continuous body. 
If we view the formation as a discrete elastic body, then the formation control problem can be viewed as a physics-based simulation in computer graphics, where the goal is to drive the body to a stress-free state defined by the elastic model.
This key insight enables us to borrow various material models from continuum mechanics to design formation controllers with different invariance properties.
In the above, we have presented several specific types of distortion energy functions that preserve shapes.
Different types of elastic models, e.g., linear elasticity model, hyperelastic model, viscoelastic model, can be defined based on the deformation gradient to characterize different material properties \cite{bonet_wood_1997}.
These models can be used similarly to guide the design of flexible/elastic formation controllers to enhance the adaptability of the formation to complex environments, which are subject to future work.

In the controller design, each element only uses the local information within the element; neither global information on the formation nor explicit communication is needed. 
This design resembles the communication-free graph-theoretical controllers, where each agent only uses relative measurements from its neighbors, and the final position of the formation is implicitly defined by the initial configuration.
If the leaders/anchors are introduced, the final position is implicitly designed by these leaders/anchors.
Similarly, our method can be easily extended to the leader-follower formation control scenario by adding external forces to a leading element, i.e., defining the distortion energy of this element by $\Psi(\F) = \|\F - \mathbf{X}\|_F^2$, where $\mathbf{X}$ is prescribed and independent of $\F$.
On the other hand, if explicit communication is allowed, one can also design the formation to a desired global position by adding external forces to all elements, i.e., defining the distortion energy of each element by $\Psi(\F) = \|\F - \mathbf{X}\|_F^2$, where $\mathbf{X}$ is independent of $\F$ which can be estimated by a consensus algorithm through communication in a distributed manner.
The above formulation can be viewed as a coupling via a common external force, which can be interpreted as a field force. 
Another method of coupling can be added via strain gradients, by adding $\|\F_{i} - \F_{j}\|_F^2$, where $\F_{i}$ and $\F_{j}$ are the deformation gradients of neighboring elements.


\section{Theoretical Analysis} \label{sec:theo}
In this section, we shall first provide some convergence analysis of the proposed controllers and discuss some features for specific controllers, and then we draw some connections of the proposed framework to existing methods, including both rigidity-based and Laplacian-based methods.

\subsection{Properties of Convergence} \label{subsec:conv}
The proposed element-based controllers essentially are gradient-based controllers minimizing the total distortion energy defined on the formation, and hence the analysis of the stability and convergence properties critically depends on their convexity, which resembles rigidity-theory-based controllers, where the analysis mainly relies on the rigidity properties of the underlying graph. 
We formally state the main results as follows. 

\begin{lemma} \label{lem:Psi}
    If Assumptions \ref{ass:non-degenerate} and \ref{ass:topo} hold, $\Psi_{(\cdot)}(\F_{e}^{*}) = 0$, $\forall e \in \mathcal{E}$, implies $\p \in \M_{(\cdot)}$ with $\M_{(\cdot)}$ being defined in \eqref{eq:M} and $(\cdot) \in \{\mathrm{T, TR, TS, TRS}\}$.
\end{lemma}
\begin{proof}
By definition of $\F_{e}$, one has $\p_{i} = \F_{e}^{*} \q_i + \mathbf{t}_{e}, \forall i \in \mathcal{V}_{e}$.
Consider any two adjacent elements $e$ and $e'$ in the connected dual graph $\hat{\mathcal{G}}$. 
They share a $(d-1)$-dimensional face composed of $d$ vertices. Let the set of these shared vertices be $\mathcal{V}_{e} \cap \mathcal{V}_{e'} = \{k_1, k_2, \dots, k_d\}$, one has:
\begin{equation}\label{eq:shared_vertex}
    \F_e^* \mathbf{q}_{k_m} + \mathbf{t}_e = \F_{e'}^* \mathbf{q}_{k_m} + \mathbf{t}_{e'}, k_m \in \mathcal{V}_{e}  \cap \mathcal{V}_{e'}.
\end{equation}
Choosing $k_1$ as the local origin and subtracting the equation of $k_1$ from the equations of the remaining $d-1$ vertices ($m = 2, \dots, d$), we eliminate the translation vectors:
\begin{equation*}
    (\F_e^* - \F_{e'}^*) (\mathbf{q}_{k_m} - \mathbf{q}_{k_1}) = \mathbf{0}, \quad m = 2, \dots, d.
\end{equation*}

By Assumption \ref{ass:non-degenerate}, the reference vertices are affinely independent, which implies that the $d-1$ vectors $\{\mathbf{q}_{k_m} - \mathbf{q}_{k_1}\}_{m=2,\dots,d}$ are linearly independent and span a $(d-1)$-dimensional hyperplane $\mathcal{P}$. 
Therefore, the two matrices $\F_e^*$ and $\F_{e'}^*$ act identically on this hyperplane: $\F_e^* \mathbf{x} = \F_{e'}^* \mathbf{x}, \forall \mathbf{x} \in \mathcal{P}$. 
We now evaluate this condition for each specific transformation group:

\begin{itemize}
    \item Translation ($\Psi_{\mathrm{T}}$): $\F_e^* = \F_{e'}^* = \matbf{I}$ is trivially satisfied.
    \item Rotation ($\Psi_{\mathrm{TR}}$): Let $\F_e^* = \mathbf{R}_e$ and $\F_{e'}^* = \mathbf{R}_{e'} \in \SO(d)$. The identical action on $\mathcal{P}$ implies $\mathbf{R}_{e'}^{-1} \mathbf{R}_e \mathbf{x} = \mathbf{x}$. 
    This indicates that the composite matrix $\tilde{\mathbf{R}} = \mathbf{R}_{e'}^{-1} \mathbf{R}_e \in \SO(d)$ fixes a $(d-1)$-dimensional subspace. 
    In $\SO(d)$, the only special orthogonal matrix that fixes a hyperplane without reflection is the identity matrix $\mathbf{I}$. Thus, $\tilde{\mathbf{R}} = \mathbf{I}$, which implies $\mathbf{R}_e = \mathbf{R}_{e'}$, and consequently $\F_e^* = \F_{e'}^*$.
    \item Scaling ($\Psi_{\mathrm{TS}}$): Let $\F_e^* = s_e \mathbf{I}$ and $\F_{e'}^* = s_{e'} \mathbf{I}$. We have $(s_e - s_{e'})\mathbf{x} = \mathbf{0}$. 
    It follows from $\mathbf{x} \neq \mathbf{0}$ that $s_e = s_{e'}$, and hence $\F_e^* = \F_{e'}^*$.
    \item Similarity ($\Psi_{\mathrm{TRS}}$): Let $\F_e^* = s_e \mathbf{R}_e$ and $\F_{e'}^* = s_{e'} \mathbf{R}_{e'}$. 
    For $\mathbf{x} \in \mathcal{P}$, taking the Euclidean norm yields $|s_e| \|\mathbf{x}\| = |s_{e'}| \|\mathbf{x}\|$. 
    Since scale factors are strictly positive ($s > 0$), we obtain $s_e = s_{e'}$. 
    Canceling the scalar degenerates the problem to the pure rotation case, yielding $\mathbf{R}_e = \mathbf{R}_{e'}$. Hence, $\F_e^* = \F_{e'}^*$.
\end{itemize}

Having rigorously established $\F_e^* = \F_{e'}^*$, substituting it back into \eqref{eq:shared_vertex} immediately yields $\mathbf{t}_e = \mathbf{t}_{e'}$. 
Since the dual graph $\hat{\mathcal{G}}$ is connected (Assumption \ref{ass:topo}), this local equality $(\F_e^*, \mathbf{t}_e) = (\F_{e'}^*, \mathbf{t}_{e'})$ propagates transitively across all elements. 
Consequently, there exists a globally uniform transformation $\F^*$ and translation $\mathbf{t}$ such that $\mathbf{p}_i = \F^* \mathbf{q}_i + \mathbf{t}$ for all $i \in \mathcal{V}$. 
Comparing this global representation with the definitions of the target manifolds in \eqref{eq:M} concludes the proof. 
\end{proof}

Lemma \ref{lem:Psi} implies that it suffices to drive $V(\p) = 0$ to complete the formation task. 
The following result establishes the property of convergence under various controllers derived using different energy functions, with the proof being provided in Appendix \ref{app:main_thm}.

\begin{theorem}
    \label{thm:main_thm}
    Under Assumptions \ref{ass:non-degenerate} and \ref{ass:topo}, consider the multi-agent system with dynamics \eqref{eq:dyn} under the control law \eqref{eq:control_law}. 
    \begin{enumerate}
        \item If $\Psi = \Psi_{\mathrm{T}}$, then the system converges to $\M_{\mathrm{T}}$ globally.
        \item If $\Psi = \Psi_{\mathrm{TR}}$, then the system converges to $\M_{\mathrm{TR}}$ locally.
        \item If $\Psi = \Psi_{\mathrm{TS}}$, then the system converges to $\M_{\mathrm{TS}}$ globally.
        \item If $\Psi = \Psi_{\mathrm{TRS}}$, then the system converges to $\M_{\mathrm{TRS}}$ locally.
    \end{enumerate}
\end{theorem}

\begin{remark}
    In the above, we employ single-integrator dynamics \eqref{eq:dyn} to transparently illustrate the core mechanism of the element-based deformation framework, where the control law behaves as a pure gradient descent. 
    However, because the total distortion energy $V(\p)$ acts as a well-defined artificial potential field, this framework can be extended to agents with more complex dynamics. 
    For instance, extending to double-integrator models requires the addition of a velocity consensus or damping term to dissipate kinetic energy. 
    Furthermore, for general linear or nonlinear systems, the proposed deformation gradient $\nabla_{\p} V(\p)$ can be integrated as the spatial coordination term within a hierarchical or passivity-based control architecture. 
    We leave the rigorous stability analysis for these general dynamics to future work.
\end{remark}

\subsection{Features of Implementation}
In this section, we present some important properties of the proposed deformation-based formation control framework, including centroid invariance and coordinate-free implementation, which resemble the well-known properties of rigidity-based formation control. 
The results are formally stated as follows.

\begin{lemma}[Centroid Invariance] \label{lem:centroid}
    Consider the multi-agent system with dynamics \eqref{eq:dyn} under the control law \eqref{eq:control_law}, where the energy $V(\p) = \sum_e w_e \Psi(\F_e)$ depends solely on the local deformation gradients. 
    The centroid of the formation, defined by $\bar{\p} = \frac{1}{N} \sum_{i=1}^N \p_i$, is invariant under the system evolution, i.e., $\dot{\bar{\p}}(t) = \mathbf{0}$ for all $t \ge 0$.
\end{lemma}
\begin{proof}
        The deformation gradient $\F_k$ depends exclusively on relative position vectors $\p_j - \p_i$. 
    Consequently, the total potential energy $V(\p)$ is invariant under global translation: $V(\p + \mathbf{1}_N \otimes \mathbf{t}) = V(\p)$ for any $\mathbf{t} \in \mathbb{R}^d$.
    Differentiating with respect to $\mathbf{t}$, we obtain:
    \begin{equation*}
        \nabla_{\mathbf{t}} V(\p + \mathbf{1} \otimes \mathbf{t}) \Big|_{\mathbf{t}=\mathbf{0}} = \sum_{i=1}^N \nabla_{\p_i} V(\p) = \mathbf{0}.
    \end{equation*}
    Since the control input is $\mathbf{u}_i = -\nabla_{\p_i} V(\p)$, the sum of control inputs is zero: $\sum_{i=1}^N \mathbf{u}_i = \mathbf{0}$.
    The dynamics of the centroid are given by:
    \begin{equation*}
        \dot{\bar{\p}} = \frac{1}{N} \sum_{i=1}^N \dot{\p}_i = \frac{1}{N} \sum_{i=1}^N \mathbf{u}_i = \mathbf{0}.
    \end{equation*}
    Thus, the centroid remains stationary throughout the process.
\end{proof}

This property confirms that the proposed deformation-based forces are pure internal stresses that reshape the formation without inducing net global motion.
In addition to the centroid invariance, the proposed framework also features a coordinate-free implementation for some of our designs.

\begin{lemma}[Coordinate-Free Implementation] \label{lem:coordinate}
    The control law \eqref{eq:control_law} with energies being defined by \eqref{eq:energy_tr} and \eqref{eq:energy_trs} can be implemented using only local relative measurements expressed in each agent's local body frame, without knowledge of a common global frame or compass.
\end{lemma}
\begin{proof}
        Let $\mathbf{Q}_i \in \SO(d)$ denote the rotation matrix from the global frame $\Sigma_g$ to the local body frame $\Sigma_i$ of agent $i$. 
    For any neighbor $j$, the relative position measured locally by agent $i$ is $\p_{ij}^{(i)} = \mathbf{Q}_i (\p_j - \p_i)$.
    
    Consider an element $e$ where agent $i$ acts as the local root. 
    The current shape matrix constructed in the local frame is:
    \begin{equation*}
        \mathbf{S}_{e}^{(i)} = [\p_{ij_1}^{(i)}, \dots, \p_{ij_d}^{(i)}] = \mathbf{Q}_i \mathbf{S}_{e},
    \end{equation*}
    where $\mathbf{S}_{e}$ is the shape matrix in the global frame.
    The local deformation gradient is computed as:
    \begin{equation*}
        \F_e^{(i)} = \mathbf{S}_{e}^{(i)} \mathbf{S}_{e,\mathrm{ref}}^{-1} = \mathbf{Q}_i \mathbf{S}_{e} \mathbf{S}_{e,\mathrm{ref}}^{-1} = \mathbf{Q}_i \F_e.
    \end{equation*}
    
    We now verify the isotropicity of the stress tensor $\mathbf{P} = \frac{\partial \Psi}{\partial \F}$. 
    A scalar energy function $\Psi(\F)$ is isotropic if $\Psi(\mathbf{R}\F) = \Psi(\F)$ for any rotation $\mathbf{R}$.
    \begin{itemize}
        \item For $\Psi_{\mathrm{TR}}(\F) = \min_{\mathbf{R}} \|\F-\mathbf{R}\|_F^2$, let $\mathbf{R}^*$ be the optimal rotation for $\F$. Then $\mathbf{Q}_i \mathbf{R}^*$ is the optimal rotation for $\mathbf{Q}_i \F$, yielding the same energy value.
        \item Similarly for $\Psi_{\mathrm{TRS}}$, the scaling factor $s$ is rotation-invariant, and the rotation part aligns automatically.
    \end{itemize}
    
    For an isotropic energy function, the first Piola-Kirchhoff stress tensor transforms covariantly:
    \begin{equation*}
        \mathbf{P}(\F^{(i)}) = \mathbf{P}(\mathbf{Q}_i \F) = \mathbf{Q}_i \mathbf{P}(\F).
    \end{equation*}
    Specifically:
    \begin{itemize}
        \item $\Psi_{\mathrm{TR}}$: $\mathbf{P}^{(i)} = 2(\F^{(i)} - \mathbf{R}^{(i)}) = 2(\mathbf{Q}_i \F - \mathbf{Q}_i \mathbf{R}) = \mathbf{Q}_i \mathbf{P}$.
        \item $\Psi_{\mathrm{TRS}}$: $\mathbf{P}^{(i)} = 2(\F^{(i)} - s\mathbf{R}^{(i)}) = \mathbf{Q}_i \mathbf{P}$.     
    \end{itemize}
    
    The control force acting on agent $i$ derived from element $e$ is given by $\mathbf{f}_i = - \mathbf{P} \mathbf{S}_{e,\mathrm{ref}}^{-\top} \mathbf{1}$ (summing contributions from edges). In the local frame:
    \begin{equation*}
        \mathbf{f}_i^{(i)} = - \mathbf{P}^{(i)} \mathbf{S}_{e,\mathrm{ref}}^{-\top} \mathbf{1} = - (\mathbf{Q}_i \mathbf{P}) \mathbf{S}_{e,\mathrm{ref}}^{-\top} \mathbf{1} = \mathbf{Q}_i \mathbf{f}_i.
    \end{equation*}
    This implies that the force vector calculated locally is exactly the global force vector represented in the agent's local frame. Since the agent's actuators (e.g., thrusters, wheels) operate in the local frame, $\mathbf{u}_i = \mathbf{f}_i^{(i)}$ can be directly applied without estimating $\mathbf{Q}_i$. Thus, the implementation is coordinate-free.
\end{proof}

\subsection{Connections with Existing Methods} \label{subsec:con}
In this section, we shall show some connections of our proposed element-based controllers with existing methods.
First, we present the connections in terms of graphical conditions.
\begin{lemma} \label{lem:inf_rigid}
    Let Assumptions \ref{ass:non-degenerate} and \ref{ass:topo} hold. 
    The graph $\mathcal{G}$ is connected. 
    The framework is infinitesimally distance/bearing/angle/RoD rigid. 
\end{lemma}
\begin{proof}
    The first statement follows from the intra-element (by definition) and the inter-element connectivity (by Assumption \ref{ass:topo}).
    The second statement follows from the first type of Henneberg construction~\cite{zhao2019bearing,jing2019angle, cao2019ratio}.
\end{proof} 

This lemma suggests that Assumption \ref{ass:topo} is a slightly stronger condition since it aims to unify different methods together.
This stems from the fact that infinitesimal rigidity with angle/RoD constraints requires more stringent graphical conditions than distance/bearing-based rigidity. Specifically, while a systematic construction involving adding vertices to neighbors has been established for the former \cite{cao2019ratio,jing2019angle}, it remains unclear whether other operations, i.e., adding vertices to non-neighbors or the second type of Henneberg construction, preserve rigidity in this context.

Next, we show that the deformation gradient $\F_e$ serves as a unified descriptor, and limiting $\F_e$ to specific manifolds naturally recovers the constraints of existing rigidity-based methods.  

\subsubsection{Connection to Displacement-based Control}
A distributed displacement-based controller aims to drive the relative position vector $\p_j - \p_i$ to a desired vector $\q_j - \q_i$, which corresponds to the translation-invariant mode ($\Psi_{\mathrm{T}}$) in our framework, where $\F_e$ is constrained to the identity matrix $\mathbf{I}$.

\begin{proposition}\label{prop:disp}
    Minimizing the translation-invariant energy $\Psi_{\mathrm{T}}(\F_e)$ is mathematically equivalent to minimizing the weighted sum of squared displacement errors for all edges in the element. Specifically:
    \begin{equation}
        \Psi_{\mathrm{T}}(\F_e) = \sum_{j=1}^d \sum_{k=1}^d (\gbf{\Gamma}_e)_{jk} (\mathbf{\delta}_j^\top \mathbf{\delta}_k),
    \end{equation}
    where $\mathbf{\delta}_k = (\p_{e_k} - \p_{e_0}) - (\q_{e_k} - \q_{e_0})$ is the displacement error of the $k$-th edge relative to the root $e_0$, and $\gbf{\Gamma}_e = (\mathbf{S}_{e,\mathrm{ref}}^\top \mathbf{S}_{e,\mathrm{ref}})^{-1}$ is the shape stiffness matrix determined by the reference configuration.
\end{proposition}
\begin{proof}
    By definition, the energy is given by:
    \begin{equation*}
        \|\mathbf{S}_e \mathbf{S}_{e,\mathrm{ref}}^{-1} - \mathbf{I}\|_F^2 = \|(\mathbf{S}_e - \mathbf{S}_{e,\mathrm{ref}}) \mathbf{S}_{e,\mathrm{ref}}^{-1}\|_F^2.
    \end{equation*}
    The term $\mathbf{\Delta} = \mathbf{S}_e - \mathbf{S}_{e,\mathrm{ref}}$ collects the displacement errors $\mathbf{\delta}_k$ as its columns. 
    Using the trace property $\|\mathbf{A}\mathbf{B}\|_F^2 = \trace(\mathbf{B}^\top \mathbf{A}^\top \mathbf{A} \mathbf{B})$:
    \begin{equation*}
        \Psi_{\mathrm{T}} = \trace( \mathbf{S}_{e,\mathrm{ref}}^{-\top} \mathbf{\Delta}^\top \mathbf{\Delta} \mathbf{S}_{e,\mathrm{ref}}^{-1} ) = \trace( \mathbf{\Delta} (\mathbf{S}_{e,\mathrm{ref}} \mathbf{S}_{e,\mathrm{ref}}^\top)^{-1} \mathbf{\Delta}^\top ).
    \end{equation*}
    Note that $(\mathbf{S}_{e,\mathrm{ref}}^\top \mathbf{S}_{e,\mathrm{ref}})^{-1}$ acts as a metric tensor derived from the reference geometry (akin to the stiffness/stress matrix).
\end{proof}

\subsubsection{Connection to Distance-based Control}
Distance-based control maintains inter-agent distances $\|\p_j - \p_i\| = \|\q_j - \q_i\|$, allowing for global rotation. 
This corresponds to the rotation-invariant mode ($\Psi_{\mathrm{TR}}$), where $\F_e \in \SO(d)$.

\begin{proposition}\label{prop:dist}
    Minimizing the rotation-invariant energy $\Psi_{\mathrm{TR}}(\F_e)$ enforces the convergence of the Green-Lagrange strain tensor $\matbf{G}_e = \frac{1}{2}(\F_e^\top \F_e - \mathbf{I})$ to zero. The squared distance error for any edge $\q_{ij}$ in the element is a projection of this tensor:
    \begin{equation}
        \|\F_e \q_{ij}\|^2 - \|\q_{ij}\|^2 = 2 \q_{ij}^\top \matbf{G}_e \q_{ij}.
    \end{equation}
\end{proposition}
\begin{proof}
    By the definition of the deformation gradient $\F_e$, the corresponding edge vector in the current configuration is given by the linear mapping:
    \begin{equation*}
         \p_{ij} = \F_e \q_{ij}.
    \end{equation*}
    The squared distance error for this edge is defined as the difference between the squared lengths of the current and reference vectors:
    \begin{equation*}
    \begin{aligned}
        \delta_{ij} &= \|\p_{ij}\|^2 - \|\q_{ij}\|^2 \\
                 &= (\F_e \q_{ij})^\top (\F_e \q_{ij}) - \q_{ij}^\top \q_{ij} \\
                 &= \q_{ij}^\top (\F_e^\top \F_e - \mathbf{I}) \q_{ij} \\
                 &= 2 \q_{ij}^\top \matbf{G}_e \q_{ij},
    \end{aligned}        
    \end{equation*}
    where $\matbf{G}_e$ denotes the Green-Lagrange strain tensor in continuum mechanics.

\end{proof}
The proposed controller minimizes the potential energy $\Psi_{\mathrm{TR}}(\F_e) = \|\F_e - \Proj_{\SO(d)}(\F_e)\|_F^2$.
    When the energy converges to zero, we have $\F_e = \mathbf{R}$ for some rotation matrix $\mathbf{R} \in \SO(d)$.
    This implies $\F_e^\top \F_e = \mathbf{R}^\top \mathbf{R} = \mathbf{I}$, and consequently, the strain tensor vanishes: $\matbf{G}_e = \mathbf{0}$.
    Since the strain tensor $\matbf{G}_e$ becomes zero, its projection along any edge direction $\q_{ij}$ also becomes zero, ensuring that the distance error $\delta_{ij}$ converges to zero for all edges within the element.
    Distance control penalizes the projection of strain along edges, whereas our method penalizes the full strain tensor norm $\|\matbf{G}_e\|_F$. 
    Thus, distance control is a sparse sampling of the deformation energy.
To provide further intuition, traditional distance-based controllers explicitly constrain only a specific set of 1D edges, requiring complex graphical rigidity conditions to prevent internal structural flexing. 
In contrast, by minimizing the full Frobenius norm of the strain tensor $\|\matbf{G}_e\|_F$, our element-based framework views the formation as a solid continuous body. 
Consequently, it inherently constrains the distance between any two arbitrary points within the element, effectively enforcing distance rigidity without relying on edge-specific graph topologies.
As will be shown in the simulation section, this property leads to a more uniform convergence rate across arbitrary desired configurations. 
However, its rigorous analysis is still an open problem since the involved tensor is trajectory-dependent and hence it is difficult to characterize the transient performance.

\subsubsection{Connection to Bearing-based Control}
Bearing-based control typically requires the direction of the edges to be preserved (parallel rigidity), i.e., $\p_j - \p_i$ is parallel to $\q_j - \q_i$, allowing for scaling. 
This corresponds to the scaling-invariant mode ($\Psi_{\mathrm{TS}}$), where $\F_e$ is a scaled identity matrix $s\mathbf{I}$.

\begin{proposition}\label{prop:bearing}
    Minimizing $\Psi_{\mathrm{TS}}(\F_e)$ enforces the convergence of the deviatoric deformation gradient $\mathrm{dev}(\F_e)$ to zero. 
    The bearing error for an edge $\q_{ij}$ is the projection of $\mathrm{dev}(\F_e) \q_{ij}$ onto the subspace orthogonal to $\q_{ij}$. Specifically:
    \begin{equation}
        \delta_{ij} = \| \mathbf{P}_{\mathbf{q}_{ij}}^\perp (\p_j - \p_i) \|^2 = \| \mathbf{P}_{\mathbf{q}_{ij}}^\perp \mathrm{dev}(\F_e)\mathbf{q}_{ij} \|^2,
    \end{equation}
    where $\mathbf{P}_{\mathbf{q}_{ij}}^\perp = \mathbf{I} - \frac{\mathbf{q}_{ij}\mathbf{q}_{ij}^\top}{\|\mathbf{q}_{ij}\|^2}$ is the orthogonal projector.
\end{proposition}
\begin{proof}
    The standard bearing error energy is defined as the squared distance of the current edge vector $\p_j - \p_i$ from the line spanned by $\q_j - \q_i$:
    \begin{equation*}
        \delta_{ij} = \| \mathbf{P}_{\mathbf{q}_{ij}}^\perp (\p_j - \p_i) \|^2 = \| \mathbf{P}_{\mathbf{q}_{ij}}^\perp \F_e \mathbf{q}_{ij} \|^2.
    \end{equation*}
    The scaling-invariant energy minimizes $\|\mathrm{dev}(\F_e)\|_F^2$.
    Substituting $\F_e = \mathrm{dev}(\F_e) + s\mathbf{I}$ into the bearing error:
    \begin{equation*}
       \begin{aligned}
        \delta_{ij} &= \| \mathbf{P}_{\mathbf{q}_{ij}}^\perp (\mathrm{dev}(\F_e) + s\mathbf{I}) \mathbf{q}_{ij} \|^2 \\
        &= \| \mathbf{P}_{\mathbf{q}_{ij}}^\perp \mathrm{dev}(\F_e) \mathbf{q}_{ij} + s \mathbf{P}_{\mathbf{q}_{ij}}^\perp \mathbf{q}_{ij} \|^2 \\
        &= \| \mathbf{P}_{\mathbf{q}_{ij}}^\perp \mathrm{dev}(\F_e) \mathbf{q}_{ij} \|^2, 
    \end{aligned}     
    \end{equation*}
    where the last equality follows from $\mathbf{P}_{\mathbf{q}_{ij}}^\perp \mathbf{q}_{ij} = \mathbf{0}$. 
\end{proof}

This equation reveals that bearing-based control penalizes the projection of the deviatoric tensor $\mathrm{dev}(\F_e)$ onto the orthogonal complement of the edge direction.
    In contrast, our proposed controller minimizes the full Frobenius norm $\|\mathrm{dev}(\F_e)\|_F^2$, effectively enforcing bearing constraints on all possible directions within the element simultaneously.
    Thus, bearing control is a sparse sampling of the deviatoric energy $\|\mathrm{dev}(\F_e)\|_F^2$ minimized by $\Psi_{\mathrm{TS}}$.
From a physical standpoint, existing methods treat geometric constraints in isolation, penalizing directional deviations edge-by-edge. 
Our framework reveals that these isolated constraints are essentially capturing specific aspects of shear deformation within the material. 
By penalizing the entire deviatoric tensor, our method simultaneously suppresses shear deformations across all possible directions, thereby intrinsically satisfying any arbitrary bearing constraints within the element.

\subsubsection{Connection to Angle and RoD Control:}
Angle-based and RoD-based controllers aim to minimize specific geometric errors invariant to similarity transformations.
We show that our similarity-invariant controller ($\Psi_{\mathrm{TRS}}$) minimizes a tensor norm that upper bounds both of these specific errors.
We define the deviatoric Cauchy-Green tensor as $\mathbf{D}_e = \mathbf{C}_e - \trace(\mathbf{C}_e)/d \cdot \mathbf{I}$, where $\mathbf{C}_e = \F_e^\top \F_e$. 
Note that $\Psi_{\mathrm{TRS}}(\F_e) \to 0$ implies $\mathbf{D}_e \to \mathbf{0}$.

\begin{proposition} \label{prop:angle}
    For any triplet of agents $(i, j, k)$ within element $e$, let $\mathbf{u} = \q_j - \q_i$ and $\mathbf{v} = \q_k - \q_i$ be the reference edge vectors, $\mathbf{g}_{ij} = \p_{ij}/\|\p_{ij}\|$, $s$ be the scale up to which the formation is close to the reference shape.
    \begin{enumerate}
        \item The angle error corresponds to the projection of $\mathbf{D}_e$ onto the off-diagonal direction defined by the edges:
        \begin{equation}
            \mathbf{g}_{ij}^\top \mathbf{g}_{ik} - \mathbf{g}_{ij}^{*\top} \mathbf{g}_{ik}^{*} \approx \frac{1}{s^2 \|\mathbf{u}\| \|\mathbf{v}\|} \mathbf{u}^\top \mathbf{D}_e \mathbf{v}.
        \end{equation}
        \item The RoD error corresponds to the projection of $\mathbf{D}_e$ onto the diagonal anisotropy direction:
        \begin{equation}
            \frac{d_{ij}^2}{d_{ik}^2} - \frac{d_{ij}^{*2}}{d_{ik}^{*2}} \approx 
            \frac{\|\mathbf{u}\|^2}{s^2 \|\mathbf{v}\|^2}  \left( 
                \frac{\mathbf{u}^\top \mathbf{D}_e \mathbf{u}}{\|\mathbf{u}\|^2} -  
                \frac{\mathbf{v}^\top \mathbf{D}_e \mathbf{v}}{\|\mathbf{v}\|^2} \right).
        \end{equation}
    \end{enumerate}
    The term in the parenthesis is exactly the inner product $\langle \mathbf{D}_e, \frac{\mathbf{u}\mathbf{u}^\top}{\|\mathbf{u}\|^2} - \frac{\mathbf{v}\mathbf{v}^\top}{\|\mathbf{v}\|^2} \rangle_F$.
    This shows that the RoD controller penalizes the anisotropy of the deformation tensor along the two edge directions, normalized by their reference lengths.

\end{proposition}
\begin{proof}
    Consider the similarity-invariant mode where the formation is close to the reference shape up to a scale $s$, i.e., $\mathbf{C}_e \approx s^2 \mathbf{I} + \mathbf{D}_e$, with $\|\mathbf{D}_e\| \ll s^2$.
    The current edge vectors are $\mathbf{u}' = \F_e \mathbf{u}$ and $\mathbf{v}' = \F_e \mathbf{v}$.
    
    \textit{1) Angle relation:}
    The term $\mathbf{g}_{ij}^\top \mathbf{g}_{ik}$ represents the cosine of the current angle $\theta'$:
    \begin{equation*}
        \cos \theta' = \frac{\mathbf{u}'^\top \mathbf{v}'}{\|\mathbf{u}'\| \|\mathbf{v}'\|} = \frac{\mathbf{u}^\top \mathbf{C}_e \mathbf{v}}{\sqrt{\mathbf{u}^\top \mathbf{C}_e \mathbf{u}} \sqrt{\mathbf{v}^\top \mathbf{C}_e \mathbf{v}}}.
    \end{equation*}
    Substituting $\mathbf{C}_e = s^2 \mathbf{I} + \mathbf{D}_e$ and using the first-order Taylor expansion $(1+x)^{-1/2} \approx 1 - x/2$:
    \begin{equation*}
        \cos \theta' \approx \frac{s^2 \mathbf{u}^\top \mathbf{v} + \mathbf{u}^\top \mathbf{D}_e \mathbf{v}}{s \|\mathbf{u}\| s \|\mathbf{v}\|} = \frac{\mathbf{u}^\top \mathbf{v}}{\|\mathbf{u}\| \|\mathbf{v}\|} + \frac{\mathbf{u}^\top \mathbf{D}_e \mathbf{v}}{s^2 \|\mathbf{u}\| \|\mathbf{v}\|}.
    \end{equation*}
    Since $\mathbf{g}_{ij}^{*\top} \mathbf{g}_{ik}^{*} = \frac{\mathbf{u}^\top \mathbf{v}}{\|\mathbf{u}\| \|\mathbf{v}\|}$, the residual is proportional $\langle \mathbf{D}_e, \mathbf{u}\mathbf{v}^\top \rangle_F$, which represents the shear component of the strain tensor.

    \textit{2) RoD Relation:}
    The current squared RoD is:
    \begin{equation*}
        \frac{d_{ij}^2}{d_{ik}^2} = \frac{\mathbf{u}^\top \mathbf{C}_e \mathbf{u}}{\mathbf{v}^\top \mathbf{C}_e \mathbf{v}} = \frac{s^2 \|\mathbf{u}\|^2 + \mathbf{u}^\top \mathbf{D}_e \mathbf{u}}{s^2 \|\mathbf{v}\|^2 + \mathbf{v}^\top \mathbf{D}_e \mathbf{v}}.
    \end{equation*}
    Using the expansion $\frac{a+c}{b+d} \approx \frac{a}{b} + \frac{c}{b} - \frac{a d}{b^2}$:
    \begin{equation*}
        \frac{d_{ij}^2}{d_{ik}^2} \approx \frac{\|\mathbf{u}\|^2}{\|\mathbf{v}\|^2} + \frac{\|\mathbf{u}\|^2}{s^2 \|\mathbf{v}\|^2}\left( 
                \frac{\mathbf{u}^\top \mathbf{D}_e \mathbf{u}}{\|\mathbf{u}\|^2} -  
                \frac{\mathbf{v}^\top \mathbf{D}_e \mathbf{v}}{\|\mathbf{v}\|^2} \right).
    \end{equation*}
    The term in the parenthesis is exactly the inner product $\langle \mathbf{D}_e, \frac{\mathbf{u}\mathbf{u}^\top}{\|\mathbf{u}\|^2} - \frac{\mathbf{v}\mathbf{v}^\top}{\|\mathbf{v}\|^2} \rangle_F$.
    This shows that the RoD controller penalizes the anisotropy of the deformation tensor along the two edge directions, normalized by their reference lengths.
\end{proof}

Thus, angle and RoD-based controllers penalize sparse scalar projections of the deviatoric tensor $\mathbf{D}_e$, whereas controller \eqref{eq:control_law} with \eqref{eq:energy_trs} penalizes the norm $\|\mathbf{D}_e\|_F^2$, thereby enforcing any arbitrary angle and RoD constraints within the element simultaneously.

Up to now, we have established the connections with existing rigidity-based controllers. 
Several remarks shall be made here:
\begin{itemize}
    \item As we have shown in Prop. \ref{prop:dist}, it suffices to minimize the Green-Lagrange strain tensor; it is also viable to use $\|\matbf{G}_{e}\|_{F}^2$ as the rotation-invariant energy function. 
The bearing-based controller discussed here still uses relative position measurements.
If a bearing-only controller were discussed \cite{zhao2019bearing}, one shall merge the normalization term into the $w_{e}$ term in \eqref{eq:energy} for a detailed analysis. 
Nevertheless, we omitted these designs and analyses for a unified and concise presentation.
\item Our current controller design requires relative position measurements within each element, which may not be directly available in some scenarios. 
However, it can be extended to work with other types of relative measurements, e.g., distances, bearings, angles, by first estimating the relative positions via distributed localization algorithms \cite{yang2024joint,chen2025multirobot,fang2023integrated,han2024survey,DBLP:conf/cdc/WangLCP25}.
We leave these extensions to future work.
\item On the other hand, from the above analysis, it can be found that, except for the displacement-based controller, existing works can be viewed as specific sparse samplings of the more general deformation energy defined on the formation. 
The underlying reason is that these methods are only given these sparse target geometric constraints.
It may be argued that our proposed controllers require more information than these existing methods, i.e., the full reference positions of all agents within each element rather than only some specific geometric constraints.
However, firstly, in practical applications, the full reference positions are prescribed manually and usually easier to be obtained; secondly, even if only sparse geometric constraints are given, our proposed controllers can still be applied by first constructing a reference formation satisfying these constraints via localization algorithms~\cite{yang2024joint,chen2025multirobot,fang2023integrated,han2024survey,DBLP:conf/cdc/WangLCP25}, since the required graphical condition has been established in Lemma \ref{lem:inf_rigid}.
\end{itemize}

Interestingly, this design workflow—reconstructing a reference geometry to define local interaction weights—closely resembles the paradigm of Laplacian-based controllers \cite{lin2016graph,lin2016necessary}. 
In those methods, the formation is implicitly defined by the kernel of a Laplacian matrix derived from the reference shape. 
This observation inspires us to bridge the two long-perceived distinct categories by examining the simplest form of our deformation energy: the Dirichlet energy.

\subsubsection{Connection to Laplacian-based Methods}
Laplacian-based formation control relies on the property that a configuration $\p$ lies in the kernel space of a constructed Laplacian matrix $\gbf{\Omega}$, i.e., $(\gbf{\Omega}\otimes \mathbf{I}) \p = \mathbf{0}$, where $\gbf{\Omega}$ can be computed from the barycentric coordinate or stress matrix \cite{lin2016distributed, zhao2018affine, yang2018constructing}.
We show that this approach can be recovered from the proposed element-based controllers, where the Laplacian matrix is explicitly derived from the reference geometry of the elements, and the control law can be interpreted as minimizing a quadratic form involving this matrix, i.e., the Dirichlet energy 
\begin{equation}
    \label{eq:Dirichlet}
    \Psi_{\text{DE}}(\F_e) = \|\F_e\|_F^2.
\end{equation}

\begin{proposition} \label{prop:lap}
    Minimizing the weighted Dirichlet energy $V(\p) = \sum_e w_e \|\F_e\|_F^2$ is exactly equivalent to minimizing the quadratic form involving a Laplacian matrix:
    \begin{equation}
        V(\p) = \trace(\p^\top (\matbf{L}\otimes \mathbf{I}) \p),
    \end{equation}
    where $\matbf{L} \in \mathbb{R}^{N \times N}$ is the global Laplacian matrix. 
    The element $\matbf{L}_{ij}$ is explicitly derived from the reference geometry of the elements sharing edge $(i,j)$.
\end{proposition}
\begin{proof}
    We derive the explicit form of $\matbf{L}$ from the local element energy.
    Consider a single element $e$ with vertices $\mathcal{V}_e = \{e_0, \dots, e_d\}$. Let $\mathbf{P}_e = [\p_{e_0}, \dots, \p_{e_d}]^\top \in \mathbb{R}^{(d+1) \times d}$ be the matrix of vertex positions.
    The current shape matrix $\mathbf{S}_e = [\p_{e_1}-\p_{i_0}, \dots, \p_{e_d}-\p_{i_0}]$ can be written as a linear transformation of $\mathbf{P}_e$:
    \begin{equation*}
        \mathbf{S}_e^\top = \mathbf{B} \mathbf{P}_e, \quad \text{where } \mathbf{B} = \begin{bmatrix} -\mathbf{1}_d & \mathbf{I}_d \end{bmatrix} \in \mathbb{R}^{d \times (d+1)}.
    \end{equation*}
    The deformation gradient is $\F_e = \mathbf{S}_e \mathbf{S}_{e,\mathrm{ref}}^{-1}$. 
    The local energy weighted by volume $w_e$ is:
    \begin{equation*}
    \begin{aligned}
        E_e & = w_e \|\F_e\|_F^2 = w_e \trace(\mathbf{S}_{e,\mathrm{ref}}^{-\top} \mathbf{S}_e^\top \mathbf{S}_e \mathbf{S}_{e,\mathrm{ref}}^{-1}) \\
        & = \trace( \mathbf{S}_e \underbrace{w_e (\mathbf{S}_{e,\mathrm{ref}}^\top \mathbf{S}_{e,\mathrm{ref}})^{-1}}_{\mathbf{K}_e} \mathbf{S}_e^\top ) \\
        & = \trace( (\mathbf{B} \mathbf{P}_e)^\top \mathbf{K}_e (\mathbf{B} \mathbf{P}_e) ) \\
        & = \trace( \mathbf{P}_e^\top \underbrace{\mathbf{B}^\top \mathbf{K}_e \mathbf{B}}_{\mathbf{L}_e} \mathbf{P}_e ).
    \end{aligned}
    \end{equation*}
    Here, $\mathbf{L}_e \in \mathbb{R}^{(d+1) \times (d+1)}$ is the \textit{local stiffness matrix} (or local stress matrix) for element $e$. 
    
    The total energy is the sum over all elements: $V(\p) = \sum_e E_e$.
    Let $\p \in \mathbb{R}^{N \times d}$ be the global position matrix. The global energy can be written as:
    \begin{equation*}
        V(\p) = \trace( \p^\top (\matbf{L}\otimes \mathbf{I}) \p ).
    \end{equation*}
    The global matrix $\matbf{L} \in \mathbb{R}^{N \times N}$ is assembled from local matrices $\mathbf{L}_e$. For any two nodes $u, v \in \{1, \dots, N\}$:
    \begin{equation*}
        \matbf{L}_{uv} = \sum_{e: u,v \in \mathcal{V}_e} (\mathbf{L}_e)_{uv}.
    \end{equation*}
    
    Thus, minimizing the Dirichlet energy is connected to the Laplacian-based formation control.
\end{proof}

It is worth noting that the unconstrained minimization of the weighted Dirichlet energy $V(\p) = \trace(\p^\top (\matbf{L}\otimes \mathbf{I}) \p)$ inevitably leads to a single-point collapse (i.e., $\p \to \vecbf{1}\otimes \vecbf{t}$).
This occurs because $\mathbf{L}$ is positive semi-definite, and the absolute minimum is achieved only when all local deformations vanish ($\F_e \to \matbf{0}$).
To prevent this trivial collapse, Laplacian-based methods require a subset of agents to act as fixed boundary conditions, known as leaders/anchors.
When leaders are present, the critical points of the Dirichlet energy form \textit{harmonic maps}, satisfying the discrete Laplace equation for the followers (i.e., $(\matbf{L}_{\myset{F}}\otimes \mathbf{I}) \p = \mathbf{0}$). 
A profound property of these harmonic maps is their response to affine boundaries: if that $(\matbf{L}_{\myset{F}} \otimes \mathbf{I}) \q = \mathbf{0}$ holds and the leaders are constrained to an affine transformation $\p_{\myset{L}} = \mathbf{A}\q_{\myset{L}} + \mathbf{b}$, the unique energy-minimizing state for the followers is to exactly replicate this affine transformation, yielding $\p_{\myset{F}} = \mathbf{A}\q_{\myset{F}} + \mathbf{b}$. 
This elegantly resolves a counterintuitive phenomenon. 
Although minimizing $\|\F_{e}\|_F^2$ ostensibly drives the local deformation toward zero, the fixed leaders physically resist this collapse. Instead, the weighted Dirichlet energy acts as a global smoothing operator. 
It penalizes non-uniform distortions, forcing the deformation to distribute evenly across the entire network. 
Consequently, the local deformation gradient $\F_e$ becomes a constant matrix $\mathbf{A}$ everywhere, seamlessly interpolating the affine boundary conditions imposed by the leaders.

It is important to note that our element-based controller design uses a strictly positive weight $w_e > 0$ for each element, which physically corresponds to pre-stressing the network with purely contractive forces.
Geometrically, this guarantees that the followers will be naturally pulled into the convex hull spanned by the leaders.
While this is sufficient if all desired follower positions lie inside this hull, it inherently prevents any follower from reaching a target position outside of it, as the continuous inward pull will inevitably drag it away from its destination. 
To address this issue, one must allow negative weights $w_e < 0$ for certain elements, representing expansive forces (or struts). 
However, this renders the energy function indefinite and may lead to instability if not designed properly. 
We shall introduce the following semi-definite programming (SDP) to find the feasible weights.

The goal is to maximize the stability margin (smallest eigenvalue of the followers' stiffness matrix) while satisfying the equilibrium conditions.
Let $\myset{V}_{\myset{L}}$ and $\myset{V}_{\myset{F}}$ denote the sets of leaders and followers, respectively. 
The global stiffness matrix $\matbf{L}$ is a linear function of the element weights $\mathbf{w} = [w_1, \dots, w_M]^\top$:
\begin{equation}
    \matbf{L}(\mathbf{w}) = \sum\nolimits_{e \in \myset{E}} w_e \matbf{L}_e,
\end{equation}
where $\matbf{L}_e \in \R^{N \times N}$ is the constant geometric stiffness basis for element $e$ (derived from the reference geometry assuming unit weight, details can be found in the proof of Prop. \ref{prop:lap}). 
Let $\matbf{L}_{{\myset{FF}}}$ be the submatrix of $\matbf{L}$ corresponding to the followers.
The optimization problem is defined as:
\begin{equation} \label{eq:sdp}
\begin{aligned}
    \underset{\mathbf{w}, \lambda}{\text{max}} \quad & \lambda - \gamma \|\mathbf{w}\|_2  \\
    \text{s.t.} \quad & \matbf{L}_{{\myset{FF}}}(\mathbf{w}) - \lambda \mathbf{I}_{|\myset{V}_{\myset{F}}|} \succeq 0,  \\
    & \sum_{j \in \mathcal{V}} \matbf{L}_{ij}(\mathbf{w}) \q_j = \mathbf{0}, \quad \forall i \in \myset{V}_{\myset{F}}, \\
    & \lambda \ge \lambda_{\min}.
\end{aligned}
\end{equation}
Another SDP formulation was presented in \cite{zhao2018affine}, ours differs from that in the construction of $\matbf{L}$ tailored for element-based graphs. 
It can be observed that the decision variables are the weights $\mathbf{w}$ and the stability margin $\lambda$.
If only a minimally connected element-based graph is used, i.e., no redundant edge between any nodes in different elements except for the shared edges between elements, the degree of freedom for the decision variables is much less than the number of constraints in \eqref{eq:sdp}, which may lead to infeasibility. 
However, the connectivity condition for element-based graphs and its connection with $(d+1)$-rooted/universal rigidity \cite{lin2016necessary} are unknown and subject to future works.

Up to now, we have established the connection among our proposed framework, rigidity-based, and Laplacian-based methods.
It has long been perceived that the latter two methods are two distinct categories with different underlying principles.
Our work reveals that they are indeed closely related through the deformation gradient representation under a unified perspective from continuum mechanics with different choices of energy functions.


\begin{figure}
    \centering
    \subfloat[]{
        \includegraphics[width=0.45\textwidth]{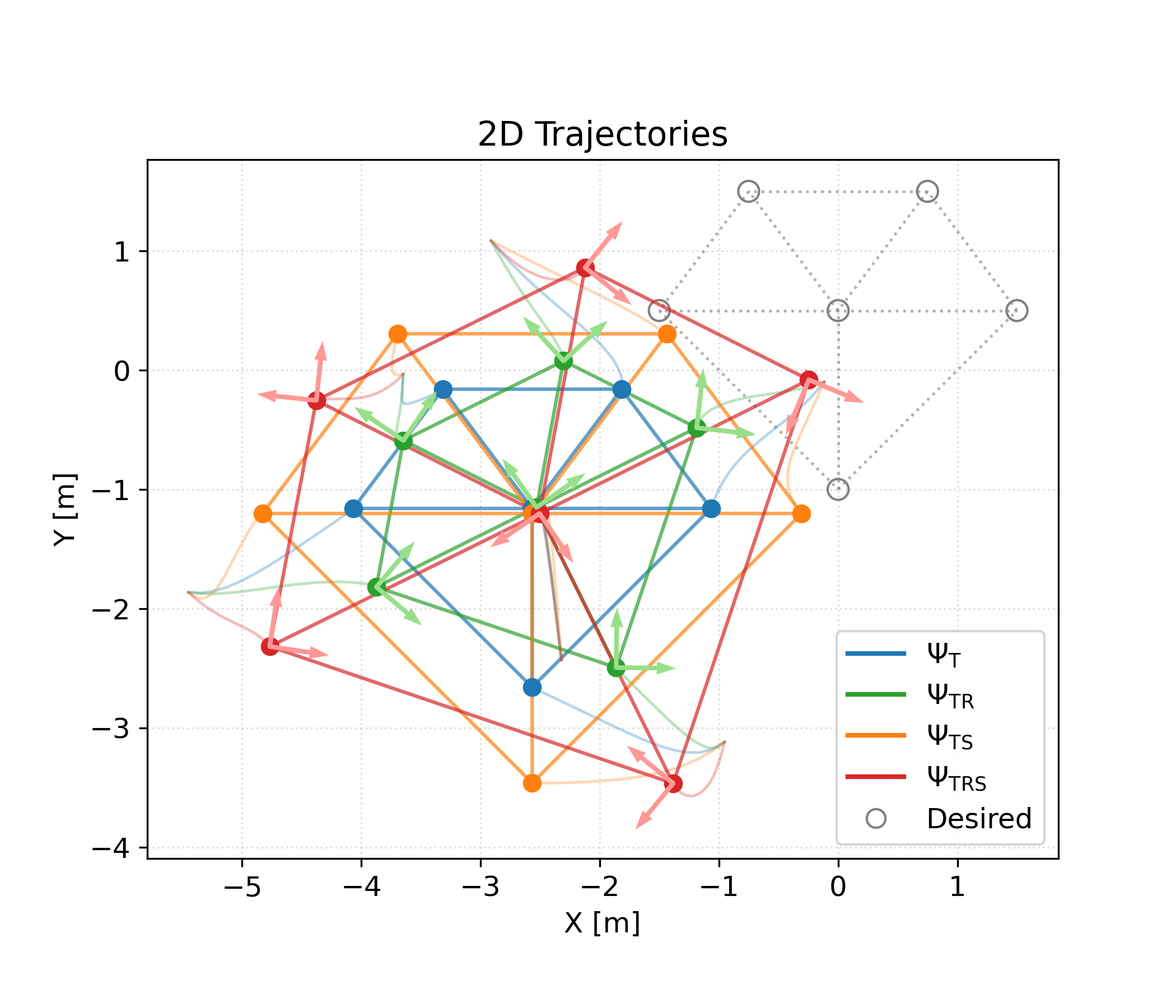}
        \label{fig:2d_traj}
    }
    \hfill
    \subfloat[]{
        \includegraphics[width=0.45\textwidth]{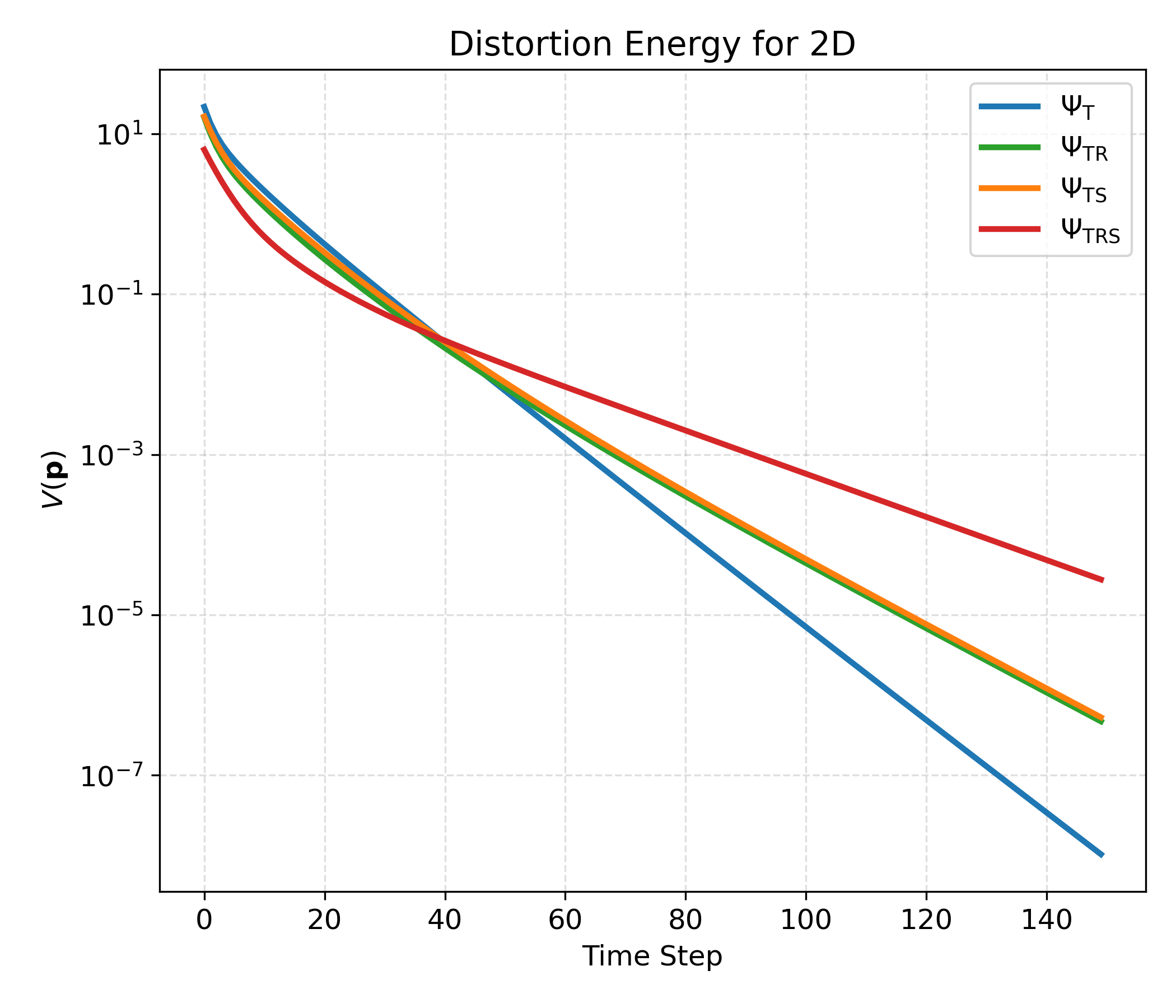}
        \label{fig:2d_energy}
    }
    \caption{Simulation results for 2D formation control ($N=6$). 
    (a) The trajectories of agents under different control laws.
    The colored arrows at the final positions of the Translation-and-Rotation-invariant and Similarity-invariant modes indicate the local coordinate frames of the agents, demonstrating that the final formation orientation is achieved without a global reference frame. 
    (b) The evolution of the potential energy $V(\p)$ on a logarithmic scale, showing exponential convergence.}
    \label{fig:2d_results}
\end{figure}

\begin{figure}[htbp]
    \centering
    \subfloat[]{
        \includegraphics[width=0.45\textwidth]{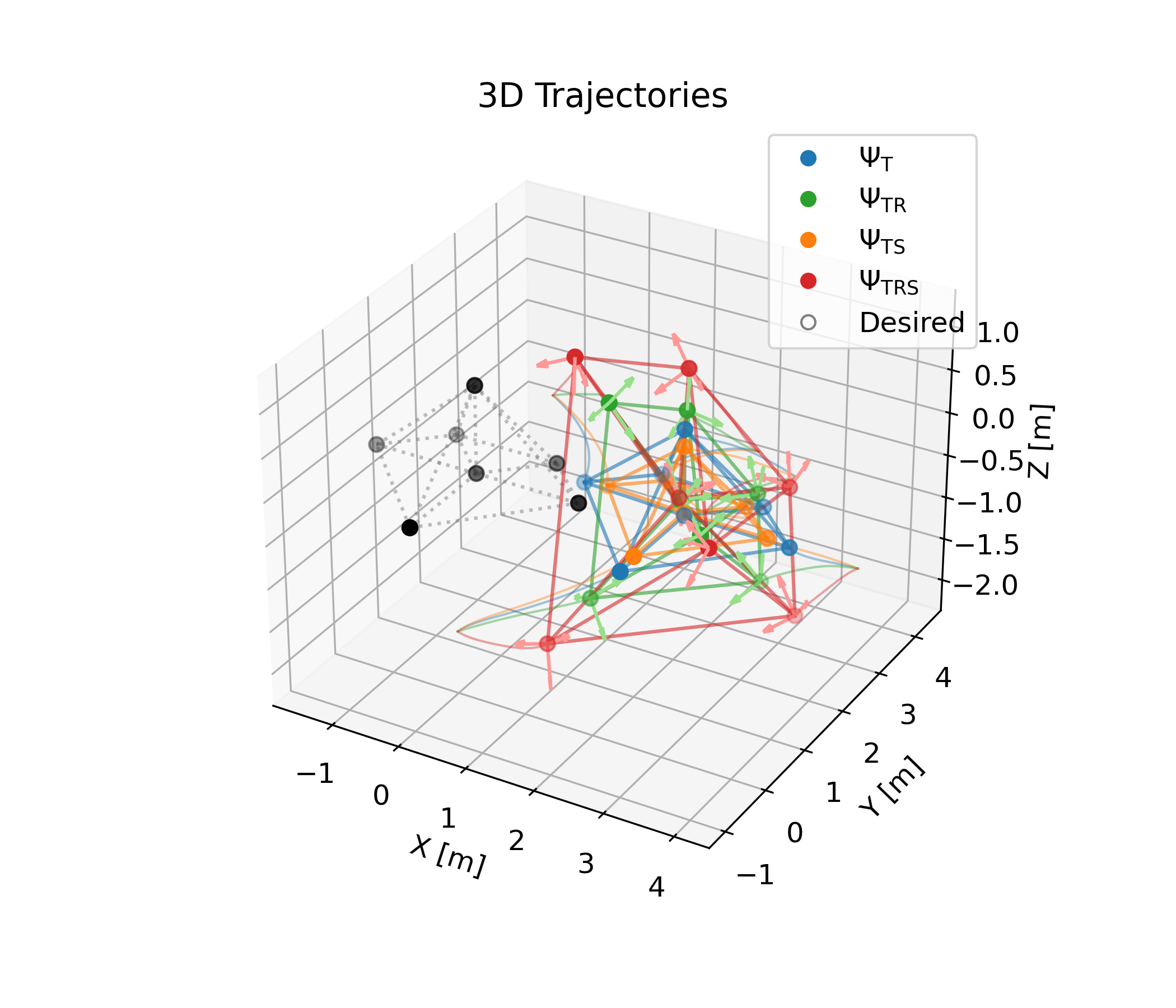}
        \label{fig:3d_traj}
    }
    \hfill
    \subfloat[]{
        \includegraphics[width=0.45\textwidth]{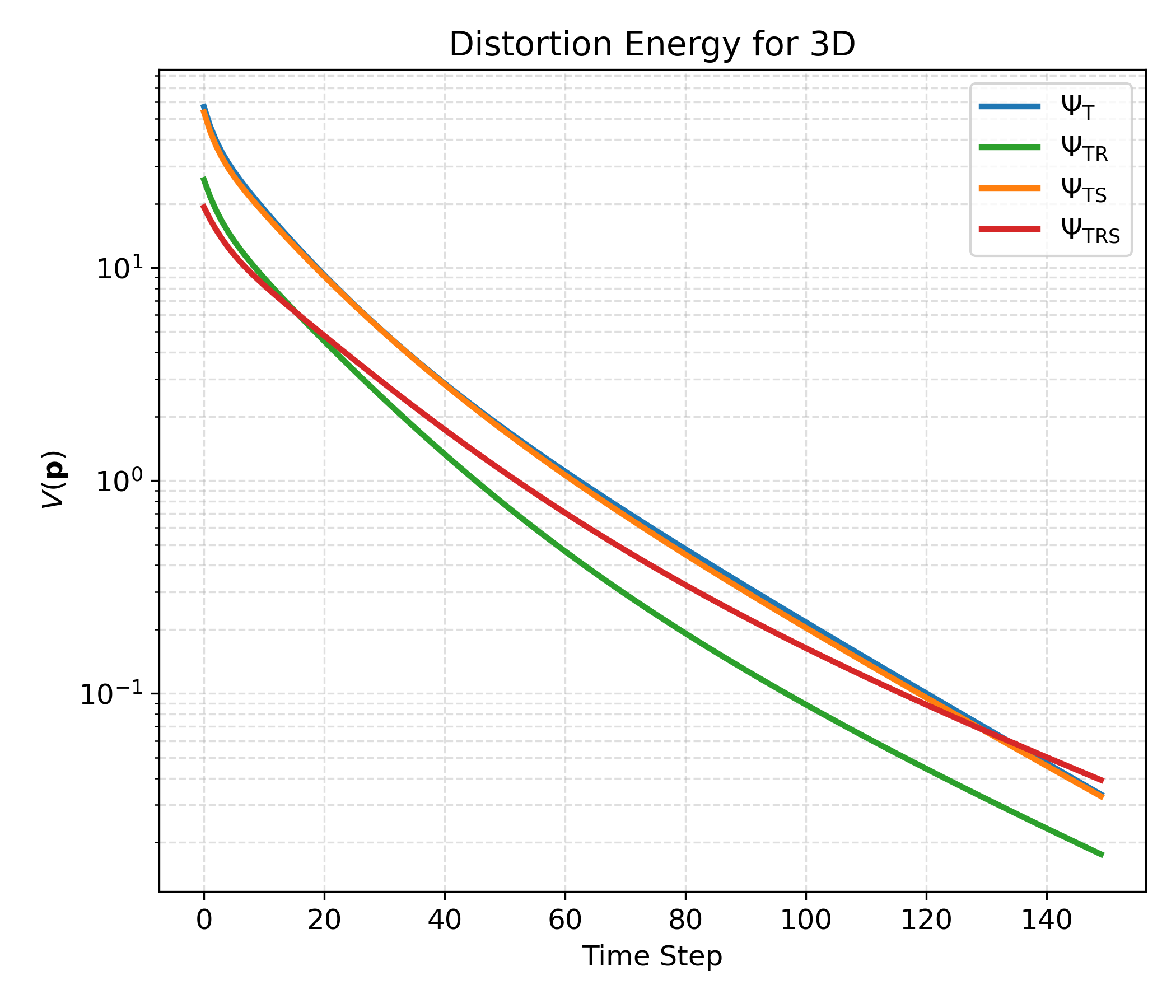}
        \label{fig:3d_energy}
    }
    \caption{Simulation results for 3D formation control ($N=7$). 
    (a) 3D trajectories of the agents. 
    The final configurations show the formation of a 3D heart-shaped structure (a pyramid with a heart base). 
    (b) The energy evolution confirms the stability and fast convergence of the proposed methods in 3D space.}
    \label{fig:3d_results}
\end{figure}

\begin{figure}
    \centering
    \includegraphics[width=\linewidth]{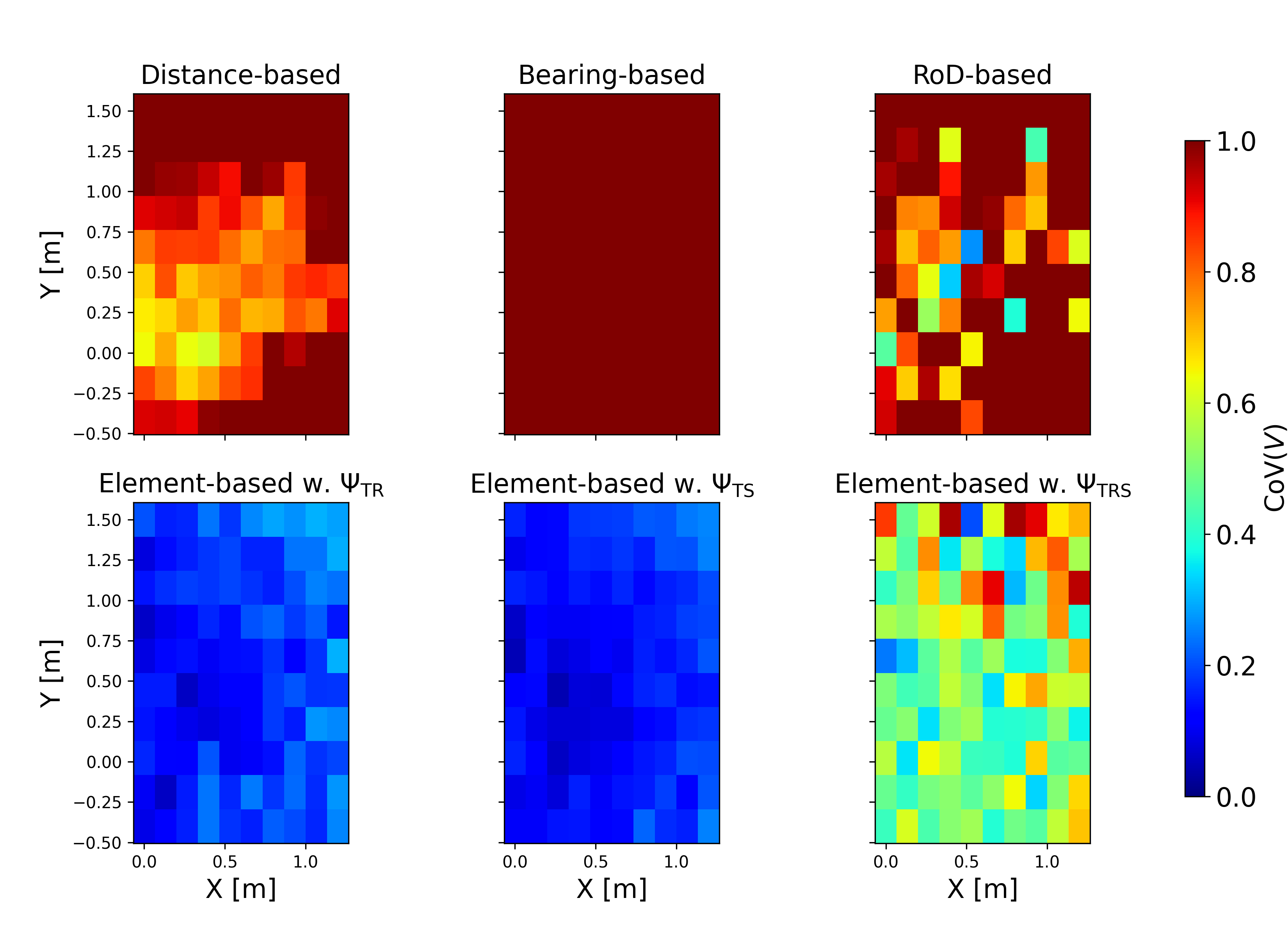}
    \caption{Heatmap of $\mathrm{CoV}(V)$ for rigidity-based (distance-, bearing-, and RoD-based) and element-based (with energy functions $\Psi_{\mathrm{TR}}$, $\Psi_{\mathrm{TS}}$, and $\Psi_{\mathrm{TRS}}$) controllers with different perturbations to the top-right node of $\q$.}
    \label{fig:cv_heatmap_comparison}
\end{figure}

\section{Simulation Results} \label{sec:simulation}
In this section, we present numerical simulations to validate the effectiveness of the proposed element-based formation control framework. 
We consider both 2D and 3D scenarios using a heart-shaped constellation of agents. 
First, we qualitatively demonstrate the geometric invariance properties and transient dynamics of the proposed controllers. 
Subsequently, we quantitatively verify their advantages to distorted configurations over traditional rigidity-based methods.

\subsection{Simulation Setup}

The interaction topology is constructed using Delaunay triangulation based on a reference configuration $\q$.
For the 2D case, we define a heart-shaped formation consisting of $N=6$ agents (one central node and five boundary nodes). 
For the 3D case, we extend the 2D shape by adding a vertex in the $z$-direction, resulting in a 3D heart structure with $N=7$ agents.

To rigorously test the convergence and the coordinate-free property, the following conditions are applied:
\begin{itemize}
    \item Initialization: The agents are initialized at positions that are a similar transformation of the reference shape, plus random Gaussian noise. 
    This ensures the initial state is far from the equilibrium.
    \item Local coordinate frame: Each agent $i$ with \eqref{eq:energy_tr} and \eqref{eq:energy_trs} operates in its own local coordinate frame $\Sigma_i$, with a randomly generated orientation $\mathbf{Q}_i \in \SO(d)$. 
    The agents do not share a common compass.
\end{itemize}

To validate the consistency of the transient convergence rates under geometrically distorted configurations, we adopt traditional rigidity-based (distance-, bearing-, and RoD-based) controllers as baselines. 
Focusing on the 2D case, we systematically perturb the desired position of the top-right node across a spatial grid of $[0, 1.2] \times [-0.4, 1.5]$. 
This specific range is selected to ensure that all evaluated methods maintain stability and avoid divergence under an identical control gain. 
To strictly isolate the intrinsic convergence behavior from the scaling influence of the control gain, we introduce the Coefficient of Variation (CoV) of the log-energy slopes as the comparative metric:
$$\mathrm{CoV}(V)=\frac{\mathrm{std}(\nabla_t \log V(\p))}{|\mathrm{mean}(\nabla_t \log V(\p))|},$$
where $\mathrm{std}(\cdot)$ and $\mathrm{mean}(\cdot)$ denote the standard deviation and mean over the active transient phase, respectively, and $\nabla_t \log V(\p)$ represents the instantaneous slope of the energy decay in the logarithmic domain.

\subsection{Results and Analysis}
The results for the 2D and 3D cases are shown in Figs.~\ref{fig:2d_results} and \ref{fig:3d_results}. 
All four controllers successfully drive the agents to form the desired heart shape, but with different final configurations reflecting their invariance groups.
As observed in Figs.~\ref{fig:2d_traj} and~\ref{fig:2d_energy}:
\begin{itemize}
    \item The translation-invariant mode (blue) forces the formation to match the reference orientation and scale exactly.
    \item The rotation-invariant mode (green) recovers the correct shape and size but settles at a rotated orientation relative to the reference.
    \item The scaling-invariant mode (orange) maintains the initial orientation but adjusts the size.
    \item The similarity-invariant mode (red) allows both rotation and scaling, finding the energy-minimal configuration closest to the initial distribution.
\end{itemize}
The arrows in Figs.~\ref{fig:2d_traj} and~\ref{fig:3d_traj} confirm the coordinate-free property: agents achieve the shape despite having random local orientations. 
Figs.~\ref{fig:2d_energy} and~\ref{fig:3d_energy} show that the energy decreases linearly on a logarithmic scale, indicating exponential convergence rates for all controllers due to the local convexity of the energy functions.

 Fig.~\ref{fig:cv_heatmap_comparison} maps the $\mathrm{CoV}$ metric across the sampled grid, illustrating the sensitivity of each controller to perturbations of the top-right node in the reference configuration $\q$.
It is evident that the traditional rigidity-based controllers exhibit vast red and yellow hotspots, indicating that they suffer from severe fluctuations and uncoordinated multi-rate convergence when the geometry becomes ill-conditioned.
In stark contrast, the proposed element-based controllers display a blue and green landscape across the entire geometric domain. 
This provides definitive empirical proof that minimizing the full deformation gradient inherently imposes isotropic tensor regularization. 
It mathematically neutralizes the spectral gap caused by geometric ill-conditioning, compelling the multi-agent system to execute highly coordinated, shape-preserving convergence regardless of the underlying structural distortion.

\section{Conclusions and Future Work} \label{sec:conclusion}

This paper has presented a unified element-based framework for multi-agent formation control, established upon the rigorous foundation of continuum mechanics. 
By shifting the fundamental unit of analysis from discrete graph edges to simplicial elements, we have shown that the deformation gradient provides a powerful and comprehensive descriptor for defining and enforcing geometric collective behaviors.
Through the lens of continuum deformation, the long-standing dichotomy between rigidity-based and Laplacian-based methods is effectively bridged.
We have demonstrated that traditional geometric constraints are essentially sparse samplings of the deformation energy tensor, while the Laplacian-based controllers emerge as the minimization of the framework's Dirichlet energy. 
This unification not only simplifies the theoretical landscape of formation control but also provides a principled methodology for designing distributed laws with various geometric invariances, supported by the physical consistency of elasticity theory.

Future work will explore the extension of this framework to more complex ``material" behaviors, such as incorporating visco-elasticity to handle dynamic obstacle avoidance or leveraging plastic deformation for adaptive formation reshaping. 
By treating robotic swarms as programmable discrete manifolds, this research opens a new pathway toward achieving sophisticated, high-dimensional coordination in autonomous systems.

\section{Appendix}
\subsection{Projection function} \label{app:proj}
For each energy function defined in Section \ref{sec:element}, the corresponding projection function $\Proj(\F)$ is given as follows.

\subsubsection{Translation-invariant energy \eqref{eq:energy_t}}
Let $\Q_{\mathrm{T}} = \{ \mathbf{I} \}$.
Since the set contains only a single element, the projection is trivial and constant:
\begin{equation*}
    \mathbf{X}^* = \Proj(\F) = \min_{\mathbf{X} \in \Q_{\mathrm{T}}} \|\F - \mathbf{X}\|_{F}^2 = \mathbf{I}.
\end{equation*}

\subsubsection{Rotation-invariant energy \eqref{eq:energy_tr}}
Let $\Q_{\mathrm{TR}} = \SO(d)$.
Next, we show that the projection is given by the polar decomposition:
\begin{equation*}
    \mathbf{X}^* = \Proj(\F) = \arg\min_{\mathbf{X} \in \Q_{\mathrm{TR}}} \|\F - \mathbf{X}\|_{F}^2 = \mathbf{R}.
\end{equation*}
We aim to solve $\min_{\mathbf{X} \in \Q_{\mathrm{TR}}} \|\F - \mathbf{X}\|_F^2$. 
Expanding the norm:
\begin{equation*}
    \|\F - \mathbf{X}\|_F^2 = \trace(\F^{\top}\F) + \trace(\mathbf{X}^{\top}\mathbf{X}) - 2\trace(\mathbf{X}^{\top}\F).
\end{equation*}
Since $\trace(\mathbf{X}^{\top}\mathbf{X}) = d$ and $\trace(\F^{\top}\F)$ are constant with respect to $\mathbf{X}$, minimizing the energy is equivalent to maximizing $\trace(\mathbf{X}^{\top}\F)$. 
This is the orthogonal Procrustes problem.
Let the SVD of $\F$ be $\F = \mathbf{U}\mathbf{\Sigma}\mathbf{V}^{\top}$. 
Then:
\begin{equation*}
    \trace(\mathbf{X}^{\top} \mathbf{U}\mathbf{\Sigma}\mathbf{V}^{\top}) = \trace(\mathbf{V}^{\top} \mathbf{X}^{\top} \mathbf{U} \mathbf{\Sigma}) = \trace(\mathbf{X}' \mathbf{\Sigma}),
\end{equation*}
where $\mathbf{X}' = \mathbf{V}^{\top} \mathbf{X}^{\top} \mathbf{U}$ is an orthogonal matrix. 
The trace is maximized when $\mathbf{X}' = \mathbf{I}$, which implies $\mathbf{V}^{\top} \mathbf{X}^{\top} \mathbf{U} = \mathbf{I}$, or equivalently, $\mathbf{X}^* = \mathbf{U}\mathbf{V}^{\top}$. 

\subsubsection{Scaling-invariant energy \eqref{eq:energy_ts}}
Let $\Q_{\mathrm{TS}} = \{ s \mathbf{I} \mid s \in \R \}$.
Next, we show that the projection is given by: 
\begin{equation*}
    \mathbf{X}^* = \Proj(\F) = \arg\min_{\mathbf{X} \in \Q_{\mathrm{TS}}} \|\F - \mathbf{X}\|_{F}^2 = s_{\mathrm{TS}} \mathbf{I}.
\end{equation*}
We seek to minimize the objective function with respect to the scalar $s$:
\begin{equation*}
    J(s) = \|\F - s\mathbf{I}\|_F^2 = \trace(\F^{\top}\F) - 2s \cdot \trace(\F) + s^2 \trace(\mathbf{I}).
\end{equation*}
Taking the derivative with respect to $s$ and setting it to zero:
\begin{equation*}
    \dv{J}{s} = -2\trace(\F) + 2s \cdot d = 0 \implies s^* = \trace(\F)/d.
\end{equation*}
Thus, the optimal projection is $\mathbf{X}^* = s_{\mathrm{TS}}\mathbf{I}$.

\subsubsection{Similarity-invariant energy \eqref{eq:energy_trs}}
Let $ \Q_{\mathrm{TRS}} = \{ s \mathbf{Q} \mid s \in \R^{+}, \mathbf{Q} \in \SO(d) \}$.
Next, we show that the projection is given by:
\begin{equation*}
    \mathbf{X}^* = \Proj(\F) = \arg\min_{\mathbf{X} \in \Q_{\mathrm{TRS}}} \|\F - \mathbf{X}\|_{F}^2 = s_{\mathrm{TRS}} \mathbf{R}.
\end{equation*}
We minimize $J(s, \mathbf{Q}) = \|\F - s\mathbf{Q}\|_F^2$.
Expanding terms:
\begin{equation*}
    J(s, \mathbf{Q}) = \trace(\F^{\top}\F) - 2s \cdot \trace(\mathbf{Q}^{\top}\F) + s^2 \trace(\mathbf{Q}^{\top}\mathbf{Q}).
\end{equation*}
Using $\trace(\mathbf{Q}^{\top}\mathbf{Q})=d$, we solve for $s$ and $\mathbf{Q}$ iteratively:
\begin{itemize}
    \item For a fixed $s > 0$, minimizing $J$ is equivalent to maximizing $\trace(\mathbf{Q}^{\top}\F)$, which yields $\mathbf{Q}^* = \mathbf{R} = \mathbf{U}\mathbf{V}^{\top}$ (same as the rotation-invariant case); 
    \item For the optimal $\mathbf{Q}^*$, we minimize with respect to $s$:
\begin{equation*}
    \frac{\partial J}{\partial s} = -2\trace(\mathbf{Q}^{*\top}\F) + 2sd = 0 \implies s^* = \frac{\trace(\mathbf{Q}^{*\top}\F)}{d}.
\end{equation*}
\end{itemize} 
Thus, the optimal projection is $\mathbf{X}^* = s^* \mathbf{Q}^* = s_{\mathrm{TRS}} \mathbf{R}$.

\subsection{Derivation of stress tensor gradient via Envelope Theorem} \label{app:grad}
In this part, we provide the detailed proof for the gradient of the unified energy function $\Psi(\F)$ with respect to the deformation gradient $\F$.

Let the energy density function be defined as the squared Frobenius distance to a constraint set $\Q$:
\begin{equation*}
    \Psi(\F) = \min_{\mathbf{Q} \in \Q} \|\F - \mathbf{Q}\|_F^2.
\end{equation*}
Let $\Proj(\F)$ denote the projection operator that maps $\F$ to the closest element in $\Q$:
\begin{equation*}
    \Proj(\F) = \arg\min_{\mathbf{Q} \in \Q} \|\F - \mathbf{Q}\|_F^2.
\end{equation*}
We can rewrite the energy as a composite function $J(\F, \mathbf{X})$, evaluated at the optimal $\mathbf{X}^* = \Proj(\F)$:
\begin{equation*}
    \Psi(\F) = J(\F, \Proj(\F)), \quad \text{where } J(\F, \mathbf{X}) = \|\F - \mathbf{X}\|_F^2.
\end{equation*}

Using the chain rule, the total derivative of $\Psi(\F)$ with respect to $\F$ is the sum of the partial derivative with respect to the first argument and the contribution from the dependence of the optimal solution on $\F$:
\begin{equation*} \label{eq:total_derivative}
    \begin{aligned}
        \frac{\partial \Psi}{\partial \F} & = \left. \frac{\partial J(\F, \mathbf{X})}{\partial \F} \right|_{\mathbf{X}=\Proj(\F)} \\
        & \quad \quad + \left. \left( \frac{\partial \Proj(\F)}{\partial \F} \right)^{\top} \odot \frac{\partial J(\F, \mathbf{X})}{\partial \mathbf{X}} \right|_{\mathbf{X}=\Proj(\F)} \\
        & = \left. 2(\F - \mathbf{X}) \right|_{\mathbf{X}=\Proj(\F)}  \\
        & \quad \quad - \left. \left( \frac{\partial \Proj(\F)}{\partial \F} \right)^{\top} \odot 2(\F - \mathbf{X}) \right|_{\mathbf{X}=\Proj(\F)} \\
        & = 2(\F - \Proj(\F)) - 2 \left( \frac{\partial \Proj(\F)}{\partial \F} \right)^{\top} \odot (\F - \Proj(\F)).
    \end{aligned}
\end{equation*}
We now show that the second term in the above is strictly zero. 
This indeed is a direct consequence of the definition of $\Proj(\F)$ as a minimizer.

Let $\mathbf{X}^* = \Proj(\F)$. 
By definition, $\mathbf{X}^*$ minimizes the objective function $J(\F, \mathbf{X})$ with respect to $\mathbf{X}$ over the manifold $\Q$. 
The first-order optimality condition states that the negative gradient of the objective function (the residual vector) must be orthogonal to the tangent space of the constraint manifold at the optimum.
Mathematically, let $\mathcal{T}_{\mathbf{X}^*} \Q$ be the tangent space of $\Q$ at $\mathbf{X}^*$. 
The optimality condition reads:
\begin{equation*} \label{eq:orthogonality}
    - \left. \frac{\partial J}{\partial \mathbf{X}} \right|_{\mathbf{X}^*} = 2(\F - \mathbf{X}^*) \perp \mathcal{T}_{\mathbf{X}^*} \mathcal{M}.
\end{equation*}
Furthermore, since $\Proj(\F)$ maps $\F$ onto the manifold $\Q$, any infinitesimal change in $\F$ results in a change $\mathrm{d}\mathbf{X}^*$ that must lie within the tangent space of the manifold. In other words, the image of the Jacobian of the projection operator lies in the tangent space:
\begin{equation*} \label{eq:tangent_condition}
    \text{range}\left( \frac{\partial \Proj(\F)}{\partial \F} \right) \subseteq \mathcal{T}_{\mathbf{X}^*} \mathcal{M}.
\end{equation*}
Combining the above, one has that their inner product (contraction) is zero.


\subsection{Proof of Theorem \ref{thm:main_thm}}
\label{app:main_thm}
    \subsubsection{Case $\Psi = \Psi_{\mathrm{T}}$}
   The total energy is given by $V(\p) = \sum_{e} w_e \|\F_e(\p) - \mathbf{I}\|_F^2$. 
   Since $\F_e(\p)$ is affine in $\p$ and the Frobenius norm is strictly convex, $V(\p)$ is convex. 
   To prove the uniqueness modulo translation, we examine the null space of its Hessian $\mathbf{H}$.
   
    Consider a perturbation $\delta \p$. The second-order variation of the energy is:
    \begin{equation*}
        \delta^2 V = \delta \p^{\top} \mathbf{H} \delta \p = 2 \sum\nolimits_{e} w_e \| \delta \F_e \|_F^2.
    \end{equation*}
    The Hessian is singular if and only if $\delta^2 V = 0$. 
    Since weights $w_e > 0$, this requires $\delta \F_e = \mathbf{0}$ for all elements $e$.
    Using the definition $\F_e = \mathbf{S}_{e} \mathbf{S}_{e,\mathrm{ref}}^{-1}$, the condition $\delta \F_e = \mathbf{0}$ implies:
    \begin{equation*}
        \delta \mathbf{S}_{e} \mathbf{S}_{e,\mathrm{ref}}^{-1} = \mathbf{0} \implies \delta \mathbf{S}_{e} = \mathbf{0},
    \end{equation*}
    since $\mathbf{S}_{e,\mathrm{ref}}$ is invertible (non-degenerate assumption).
    The condition $\delta \mathbf{S}_{e} = \mathbf{0}$ implies that for any edge $(i,j)$ within simplex $e$, the relative displacement is zero:
    \begin{equation*}
        \delta \p_j - \delta \p_i = \mathbf{0} \implies \delta \p_j = \delta \p_i.
    \end{equation*}
    This means all agents within an element must move by the same translation vector. Under the assumption that each element is connected (i.e., any two agents are connected by a path of overlapping elements), this local translation constraint propagates globally:
    \begin{equation*}
        \delta \p_1 = \delta \p_2 = \dots = \delta \p_N = \mathbf{t}, \quad \mathbf{t} \in \mathbb{R}^d.
    \end{equation*}
    Thus, the null space of $\matbf{H}$ spans exactly the subspace of global translations. 
    In the quotient space $\mathbb{R}^{dN} / \Span(\mathbf{1}_N \otimes \mathbf{I}_d)$, the Hessian is positive definite, implying strict convexity and a unique global minimum for the formation shape. 
    The statement follows from the convergence of the gradient algorithm on strictly convex functions.
    
    \subsubsection{Case $\Psi = \Psi_{\mathrm{TR}}$}
    Consider the system at a global equilibrium $\p^*$ where $\F_e = \mathbf{R}$ for all $e$. 
    Without loss of generality, let us analyze the perturbation in a local frame where $\mathbf{R} = \mathbf{I}$. 
    Let $\p = \p^* + \delta \p$. The corresponding perturbation in the deformation gradient is $\delta \F_e$.
    
    The second-order variation of the local energy $\Psi_{\mathrm{TR}}(\F_e) = \min_{\mathbf{R}} \|\F_e - \mathbf{R}\|_F^2$ near $\F_e=\mathbf{I}$ is given by the squared norm of the linearized strain tensor (symmetric part of $\delta \F_e$, see Appendix \ref{app:TR} for detailed derivation):
    \begin{equation*}
        \delta^2 \Psi_{\mathrm{TR}} \approx \| \text{sym}(\delta \F_e) \|_F^2 = \frac{1}{4} \| \delta \F_e + \delta \F_e^{\top} \|_F^2.
    \end{equation*}
    The total Hessian quadratic form is:
    \begin{equation*}
        \delta \p^{\top} \mathbf{H} \delta \p = \sum\nolimits_{e} w_e \| \text{sym}(\delta \F_e) \|_F^2.
    \end{equation*}
    The Hessian is singular if and only if the quadratic form is zero, which implies $\text{sym}(\delta \F_e) = \mathbf{0}$ for all $e$. 
    This condition $\delta \F_e + \delta \F_e^{\top} = \mathbf{0}$ means that each local deformation is an infinitesimal rotation (skew-symmetric matrix $\bm{\Omega}_e$):
    \begin{equation*}
        \delta \p_j - \delta \p_i = \bm{\Omega}_e (\q_j - \q_i), \quad \forall i,j \in \mathcal{V}_e.
    \end{equation*}
    This implies that the motion of vertices in element $e$ is a rigid body velocity field: $\delta \p_i = \bm{\Omega}_e \q_i + \mathbf{t}_e$.
    
    Consider two elements $e$ and $e'$ sharing a non-degenerate face (at least $d$ vertices in $\mathbb{R}^d$).
    The continuity of the velocity field across the shared vertices requires:
    \begin{equation*}
        \bm{\Omega}_e \q + \mathbf{t}_e = \bm{\Omega}_{e'} \q + \mathbf{t}_{e'}, \quad \forall \q \in \mathcal{V}_e \cap \mathcal{V}_{e'}.
    \end{equation*}
    Since the shared vertices span a $(d-1)$-dimensional affine subspace, this equality enforces $\bm{\Omega}_e = \bm{\Omega}_{e'}$ and $\mathbf{t}_e = \mathbf{t}_{e'}$. By the connectivity of the dual graph (Assumption 2), this constraint propagates to all elements, implying $\bm{\Omega}_e \equiv \bm{\Omega}$ and $\mathbf{t}_e \equiv \mathbf{t}$ for all $e$.
    
   \subsubsection{Case $\Psi = \Psi_{\mathrm{TS}}$}
   Similar to Case 1), the global convexity follows from the affine mapping from $\p$ to $\F_e(\p)$.
    Consider the Hessian quadratic form for a perturbation $\delta \p$:
    \begin{equation*}
        \delta^2 V = \delta \p^{\top} \mathbf{H} \delta \p = 2 \sum\nolimits_{e} w_e \| \mathrm{dev}(\delta \F_e) \|_F^2,
    \end{equation*}
    where $\mathrm{dev}(\mathbf{A}) = \mathbf{A} - \trace(\mathbf{A})/d \cdot \mathbf{I}$ denotes the deviatoric part of matrix $\mathbf{A}$.
    The Hessian is singular if and only if $\mathrm{dev}(\delta \F_e) = \mathbf{0}$ for all elements $e$.
    The condition $\mathrm{dev}(\delta \F_e) = \mathbf{0}$ implies that the local deformation gradient is a scaled identity matrix:
    \begin{equation*}
        \delta \F_e = s_e \mathbf{I}, \quad s_e \in \mathbb{R}.
    \end{equation*}
    By definition of $\F_e$, one has:
    \begin{equation*}
        \delta \p_j - \delta \p_i = s_e (\q_j - \q_i).
    \end{equation*}
    This implies that the motion of vertices in element $e$ is a uniform scaling field: $\delta \p_i = s_e \q_i + \mathbf{t}_e$.

    Now, consider two elements $e$ and $e'$ sharing a non-degenerate face (or edge). 
    The continuity of the velocity field across the shared vertices requires:
    \begin{equation*}
        s_e \q + \mathbf{t}_e = s_{e'} \q + \mathbf{t}_{e'}, \quad \forall \q \in \mathcal{V}_e \cap \mathcal{V}_{e'}.
    \end{equation*}
    Since the shared vertices span a $(d-1)$-dimensional affine subspace, this equality enforces $s_e = s_{e'}$ and $\mathbf{t}_e = \mathbf{t}_{e'}$.
    The rest of the proof follows that of Case 2). 
    
    \subsubsection{Case $\Psi = \Psi_{\mathrm{TRS}}$}
    Similar to case 2), without loss of generality, we analyze the perturbation near the equilibrium $\F = \mathbf{I}$. 
The second order variation corresponds to the squared norm of the deviatoric strain (see Appendix \ref{app:TRS} for detailed derivation):
    \begin{equation*}
        \delta^2 \Psi_{\mathrm{TRS}} \approx \| \bm{\epsilon} - \trace(\bm{\epsilon})/d \cdot \mathbf{I} \|_F^2 = \| \mathrm{dev}(\mathrm{sym}(\delta \F)) \|_F^2.
    \end{equation*}
    
    The total Hessian is singular if and only if the local deviatoric strain is zero for all elements:
    \begin{equation*}
        \mathrm{dev}(\mathrm{sym}(\delta \F_e)) = \mathbf{0} \implies \mathrm{sym}(\delta \F_e) = s_e \mathbf{I}, \quad \exists s_e \in \mathbb{R}.
    \end{equation*}
    Combining with the rotational part, the general form of the zero-energy deformation gradient perturbation is:
    \begin{equation*}
        \delta \F_e = s_e \mathbf{I} + \bm{\Omega}_e.
    \end{equation*}
    This implies the local velocity field is an infinitesimal similarity transformation:
    \begin{equation*}
        \delta \p_j - \delta \p_i = (s_e \mathbf{I} + \bm{\Omega}_e) (\q_j - \q_i),
    \end{equation*}
    which further implies that the motion of vertices in element $e$ is given by:
    \begin{equation*}
        \delta \p_i = (s_e \mathbf{I} + \bm{\Omega}_e) \q_i + \vecbf{t}_e.
    \end{equation*}
    The rest of the proof follows that of Case 2).



\subsection{Second order variation of $\Psi_{\mathrm{TR}}$}
\label{app:TR}
    We rigorously derive the quadratic approximation of $\Psi_{\mathrm{TR}}(\F) = \min_{\matbf{X} \in \SO(d)} \|\F - \matbf{X}\|_F^2$ near $\F = \mathbf{I}$.
    Let $\F = \mathbf{I} + \delta \F$. The rotation matrix $\matbf{X}$ near identity can be approximated by its first-order Taylor expansion (Lie algebra element):
    \begin{equation*}
        \matbf{X} \approx \mathbf{I} + \bm{\Omega}, \quad \text{where } \bm{\Omega}^{\top} = -\bm{\Omega} \text{ (skew-symmetric)}.
    \end{equation*}
    The minimization problem becomes:
    \begin{equation*}
    \begin{aligned}
        \Psi_{\mathrm{TR}}(\mathbf{I} + \delta \F) 
        & \approx \min_{\bm{\Omega}} \| (\mathbf{I} + \delta \F) - (\mathbf{I} + \bm{\Omega}) \|_F^2 \\
        & = \min_{\bm{\Omega}} \| \delta \F - \bm{\Omega} \|_F^2.
    \end{aligned}
    \end{equation*}
    We decompose the perturbation $\delta \F$ into its symmetric (strain) and skew-symmetric (rotation) parts:
    \begin{equation*}
        \delta \F = \underbrace{\frac{\delta \F + \delta \F^{\top}}{2}}_{\bm{\epsilon} \text{ (sym)}} + \underbrace{\frac{\delta \F - \delta \F^{\top}}{2}}_{\bm{\omega} \text{ (skew)}}.
    \end{equation*}
    Substituting this into the norm:
    \begin{equation*}
        \| \delta \F - \bm{\Omega} \|_F^2 = \| \bm{\epsilon} + (\bm{\omega} - \bm{\Omega}) \|_F^2.
    \end{equation*}
    A crucial property of the Frobenius norm is the orthogonality between symmetric and skew-symmetric matrices ($\langle \mathbf{S}, \mathbf{A} \rangle_{F} = 0$). Thus, the cross-term vanishes:
    \begin{equation*}
        \| \bm{\epsilon} + (\bm{\omega} - \bm{\Omega}) \|_F^2 = \| \bm{\epsilon} \|_F^2 + \| \bm{\omega} - \bm{\Omega} \|_F^2.
    \end{equation*}
    To minimize this energy, the optimal infinitesimal rotation is clearly $\bm{\Omega}^* = \bm{\omega}$ (which cancels out the rotational part of the perturbation). The remaining energy is purely the squared norm of the symmetric part:
    \begin{equation*}
        \Psi_{\mathrm{TR}} \approx \| \bm{\epsilon} \|_F^2 = \frac{1}{4} \| \delta \F + \delta \F^{\top} \|_F^2.
    \end{equation*}

\subsection{Second order variation of $\Psi_{\mathrm{TRS}}$}
\label{app:TRS}
    We derive the quadratic approximation of the similarity-invariant energy $\Psi_{\mathrm{TRS}}(\F) = \min_{s \in \mathbb{R}^{+}, \matbf{R} \in \SO(d)} \|\F - s\matbf{R}\|_F^2$ near $\F = \mathbf{I}$.
    Let $\F = \mathbf{I} + \delta \F$.
This energy allows both rotation and uniform scaling. 
The target manifold is the group of scaled rotation matrices $s\mathbf{R}$. 
Near the identity $\mathbf{I}$, a similarity transformation can be approximated as:
    \begin{equation*}
        s\mathbf{R} \approx (1 + \gamma)\mathbf{I} + \bm{\Omega},
    \end{equation*}
    where $\gamma \in \mathbb{R}$ represents infinitesimal scaling and $\bm{\Omega}$ represents infinitesimal rotation.
    The energy density becomes:
    \begin{equation*}
        \Psi_{\mathrm{TRS}}(\mathbf{I} + \delta \F) \approx \min_{\gamma, \bm{\Omega}} \| \delta \F - (\gamma\mathbf{I} + \bm{\Omega}) \|_F^2.
    \end{equation*}
    Decomposing $\delta \F$ into its symmetric part $\bm{\epsilon}$ and skew-symmetric part $\bm{\omega}$, the norm splits due to orthogonality:
    \begin{equation*}
        \| (\bm{\epsilon} - \gamma\mathbf{I}) + (\bm{\omega} - \bm{\Omega}) \|_F^2 = \| \bm{\epsilon} - \gamma\mathbf{I} \|_F^2 + \| \bm{\omega} - \bm{\Omega} \|_F^2.
    \end{equation*}
    The optimal rotation is still $\bm{\Omega}^* = \bm{\omega}$. The optimal scaling $\gamma^*$ minimizes the symmetric part:
    \begin{equation*}
        \min_{\gamma} \| \bm{\epsilon} - \gamma\mathbf{I} \|_F^2 \implies \gamma^* = \trace(\bm{\epsilon})/d = \trace(\delta \F)/d.
    \end{equation*}
    The residual energy corresponds to the squared norm of the deviatoric strain:
    \begin{equation*}
        \Psi_{\mathrm{TRS}} \approx \| \bm{\epsilon} - \trace(\bm{\epsilon})/d \cdot \mathbf{I} \|_F^2 = \| \text{dev}(\text{sym}(\delta \F)) \|_F^2.
    \end{equation*}


\bibliographystyle{IEEEtran}
\bibliography{ref.bib}

\end{document}